%% file: kdd.tex
\documentclass[11pt]{article}
\input{macros.tex}

\usepackage{hyperref}

\newtheorem{proof}{Proof}

\let\olddefinition\proof
\renewcommand{\proof}{\olddefinition\normalfont \unskip}

\begin{document}

\title{Scalable Betweenness Centrality Maximization via Sampling}

\author{Ahmad Mahmoody \thanks{Brown University, \href{mailto:ahmad@brown.edu}{ahmad@brown.edu} }
\and
Charalampos E. Tsourakakis\thanks{Harvard University, \href{mailto:babis@seas.harvard.edu}{babis@seas.harvard.edu} }
\and 
Eli Upfal \thanks{Brown University, \href{mailto:eli@cs.brown.edu}{eli@cs.brown.edu} }
}

\maketitle
\begin{abstract}
\input{tex/abstract.tex}

\end{abstract}

\input{tex/Introduction.tex}

\input{tex/Related.tex}

\input{tex/Proposed.tex}

\input{tex/assumption.tex}

\input{tex/Experiments.tex}

\input{tex/Conclusion.tex}

\bibliographystyle{abbrv}
\bibliography{centRefnew}

\end{document}

%% file: macros.tex
\usepackage{amsmath}
\usepackage{amsfonts}
\usepackage{amssymb,bbm}
\usepackage{latexsym}
\usepackage{graphicx}
\usepackage{color}
\usepackage{url}
\usepackage{xspace} 
\usepackage[algoruled,vlined]{algorithm2e}
\usepackage{algpseudocode}
\usepackage{adjustbox}
\usepackage{floatrow}

\usepackage[sc]{mathpazo}
\usepackage{multirow}

\usepackage{booktabs}
\usepackage{caption}
\usepackage{subcaption}

\usepackage{soul}

\usepackage{wrapfig}
\usepackage{lipsum}

\newcommand{\ya}{\texttt{Y-ALG}\xspace}

\newcommand{\BW}{\textsf{BWC}\xspace}
\newcommand{\CMP}{\textsf{CMP}\xspace}
\newcommand{\ABW}{\textsf{A-BWC}\xspace}
\newcommand{\paran}[1]{\left(#1\right)}
\newcommand*{\qed}{\hfill\ensuremath{\blacksquare}}%

 \newcommand*{\algonamebase}{HEDGE}

\newcommand*{\algoname}{\texttt{\algonamebase}\xspace}

\newcommand{\babis}[1]{\textcolor{blue}{ Babis-2-Ahmad: ``#1''}}

\newtheorem{definition}{Definition}

\newcommand{\E}[1]{\mathbb{E}\left(#1\right)}

\def\he{hyper-edge}
\def\He{Hyper-edge} 
\def\opt{\texttt{OPT}}
\DeclareMathOperator*{\argmax}{arg\,max}
\newcommand{\pr}[1]{\text{Pr}\left(#1\right)}

\def\A{\mathcal{A}}
\def\H{\mathcal{H}}

\newcommand{\field}[1]{\mathbb{#1}} 
\newtheorem{theorem}{Theorem}

\newtheorem{lemma}{Lemma}

\newcommand{\spara}[1]{\smallskip\noindent{\bf #1}\ }

\newcommand{\hide}[1]{}
 
\newcommand{\squishlist}{
 \begin{list}{$\bullet$}
  {  \setlength{\itemsep}{0pt}
     \setlength{\parsep}{3pt}
     \setlength{\topsep}{3pt}
     \setlength{\partopsep}{0pt}
     \setlength{\leftmargin}{2em}
     \setlength{\labelwidth}{1.5em}
     \setlength{\labelsep}{0.5em}
} }
\newcommand{\squishlisttight}{
 \begin{list}{$\bullet$}
  { \setlength{\itemsep}{0pt}
    \setlength{\parsep}{0pt}
    \setlength{\topsep}{0pt}
    \setlength{\partopsep}{0pt}
    \setlength{\leftmargin}{2em}
    \setlength{\labelwidth}{1.5em}
    \setlength{\labelsep}{0.5em}
} }

\newcommand{\squishdesc}{
 \begin{list}{}
  {  \setlength{\itemsep}{0pt}
     \setlength{\parsep}{3pt}
     \setlength{\topsep}{3pt}
     \setlength{\partopsep}{0pt}
     \setlength{\leftmargin}{1em}
     \setlength{\labelwidth}{1.5em}
     \setlength{\labelsep}{0.5em}
} }

\newcommand{\squishend}{
  \end{list}
}

\newcommand{\squishlistt}{
 \begin{list}{---}
  {  \setlength{\itemsep}{0pt}
     \setlength{\parsep}{3pt}
     \setlength{\topsep}{3pt}
     \setlength{\partopsep}{0pt}
     \setlength{\leftmargin}{2em}
     \setlength{\labelwidth}{1.5em}
     \setlength{\labelsep}{0.5em}
} }

\newtheorem{remark}{Remark}
\newcommand{\set}[1]{\left\{#1\right\}}

\def\he{hyper-edge}
\def\He{Hyper-edge}

\newcommand*{\exnamebase}{EXHAUST}

\newcommand*{\exname}{\texttt{\exnamebase}\xspace}

%% file: tex/abstract.tex
\emph{Betweenness centrality} (\BW) is a fundamental centrality measure in social network analysis. Given a large-scale network, how can we find   the most central nodes? This question is of key importance to numerous important applications that rely on \BW, including community detection and understanding graph vulnerability.   Despite the large amount of work on   scalable approximation algorithm design for \BW, estimating \BW on large-scale networks remains a computational challenge.

In this paper, we study the \emph{Centrality Maximization} problem (\CMP): given a graph $G=(V,E)$ and a positive integer $k$, find a set $S^* \subseteq V$  that maximizes \BW subject to the cardinality constraint  $|S^*| \leq k$. We present an efficient randomized algorithm that provides a $(1-1/e-\epsilon)$-approximation with high probability, where $\epsilon>0$.  Our results improve the current state-of-the-art result~\cite{yoshida2014almost}. Furthermore, we provide  
the first theoretical evidence for the validity of a  \emph{crucial} assumption in betweenness centrality estimation, namely that in real-world networks   $O(|V|^2)$ shortest paths 
pass through the top-$k$ central nodes, where $k$ is a constant.  This also explains why our  algorithm runs in near linear time on real-world networks. We also show that our algorithm and analysis can be applied to a wider range of centrality measures, by providing a general analytical framework.

On the experimental side, we perform an extensive experimental analysis of our method on real-world networks, demonstrate its accuracy and scalability, and study different properties of central nodes. Then, we compare the sampling method used by the state-of-the-art algorithm with our method.
Furthermore, we perform a study of \BW in time evolving networks, and see how the centrality of the central nodes in the graphs changes over time. Finally, we compare the performance of the stochastic Kronecker model \cite{leskovec2005realistic}  to real data, and  observe that it generates a similar growth pattern.

%% file: tex/Introduction.tex
\section{Introduction} 
\label{sec:intro}

Betweenness centrality (\BW) is a fundamental measure in network analysis, measuring the effectiveness of a vertex in connecting pairs of vertices via shortest paths~\cite{freeman}. Numerous graph mining applications rely on betweenness centrality, such as detecting   communities in social and biological networks \cite{girvan2002community} and  understanding the capabilities of an adversary with respect to attacking a network's connectivity \cite{iyer2013attack}. 
The betweenness centrality of a node $u$ is defined as 
$$B(u)=\sum_{s,t}\frac{\sigma_{s,t}(u)}{\sigma_{s,t}},$$
 where $\sigma_{s,t}$ is the number of $s$-$t$ shortest paths, and $\sigma_{s,t}(u)$ is the number of $s$-$t$ shortest paths that have $u$ as their internal node. However, in many  applications, e.g. \cite{girvan2002community,iyer2013attack}, we are interested in centrality of sets of nodes. For this reason, the notion of \BW has been  extended  to sets of nodes~\cite{ishakian2012framework,yoshida2014almost}. For a set of nodes $S \subseteq V$, we define the betweenness centrality of $S$ as 
 
 $$   B(S) = \sum_{s,t\in V}\frac{\sigma_{s,t}(S)}{\sigma_{s,t}},$$

\noindent where $\sigma_{s,t}(S)$ is the number of $s$-$t$ shortest paths that have an internal node in $S$. Note that we cannot obtain 
$B(S)$  from the values $\{B(v), v \in S\}$. In this work, we study the \emph{Centrality Maximization problem} (\CMP) defined formally as follows:
 \begin{definition}[\CMP]  
	~~Given a network $G=(V,E)$ and a positive integer $k$, find a subset $S^*\subseteq V$ such that 
	$$S^*\in \argmax_{S\subseteq V: |S|\leq k} B(S).$$
	We also denote the maximum centrality of a set of $k$ nodes by $\opt_k$, i.e.,  $\opt_k = \max\limits_{S\subseteq V: |S|\leq k} B(S)$.
 \end{definition}

It is known that \CMP is APX-complete~\cite{fink2011maximum}.
The best deterministic algorithms for \CMP rely on the fact that \BW is monotone-submodular and provide a $(1-1/e)$-approximation \cite{fink2011maximum,dolev2009incremental}. However, the running time of these algorithms is at least quadratic in the input size, and do not scale well to large-scale networks. 

Finding the most central nodes in a network is a computationally challenging problem that we are able to handle accurately and efficiently. In this paper we focus on scalability of \CMP, and
graph mining applications. Our main contributions are summarized as follows.

\spara{Efficient algorithm.} We provide a randomized approximation algorithm, \algoname, based on sampling shortest paths, for accurately estimating the \BW and solving \CMP. Our algorithm is simple, scales gracefully as the size of the graph grows, and improves the previous result~\cite{yoshida2014almost}, by (i) providing a $(1-1/e-\epsilon)$-approximation, and (ii) smaller sized samples. Specifically, in  Yoshida's algorithm   \cite{yoshida2014almost}, a sample contains all the nodes on ``any'' shortest path between a pair, whereas in our algorithm, each sample is just a set of nodes from a single shortest path between the pair. 

\hide{ We also show that the number of samples in~\cite{yoshida2014almost}, even with the $\opt_k = \Theta(n^2)$ assumption, is not sufficient.}

\spara{The $\opt_k=\Theta(n^2)$ assumption.} 
Prior work on \BW estimation strongly relies on the assumption that $\opt_k = \Theta(n^2)$ for a constant integer $k$~\cite{yoshida2014almost}. 
As we show, this assumption is not true in general. So far,
only  empirical evidence supports this strong assumption.

We show that two broad families of networks satisfy this assumption:  bounded treewidth networks and a popular family of stochastic networks that provably generate scale-free, small-world graphs with high probability. Note that the classical Barab{\'a}si-Albert scale-free random tree model \cite{barabasi1999emergence,mori2005maximum} belongs to the former category.  Our results imply that the $\opt_k = \Theta(n^2)$ assumption holds even for $k=1$, for these families of networks. To our knowledge, this is the first theoretical evidence for the validity of this {\em crucial} assumption on real-world networks.

\spara{General analytical framework.} To analyze our algorithm, \algoname, we provide a general analytical framework based on Chernoff bound and submodular optimization, and show that it can be applied to any other centrality measure if it (i)  is monotone-submodular, and (ii)   admits a {\he} sampler (defined in Sect.~\ref{sec:proposed}). Two examples of such centralities are the \emph{coverage}~\cite{yoshida2014almost} and the \emph{$\kappa$-path} centralities~\cite{alahakoon2011k}. 

\hide{(see Sect.~\ref{sec:beyond}). \babis{please add references here for these two centrality measures}}

\spara{Experimental evaluation.} We provide an experimental evaluation of our algorithm that shows that 
it scales gracefully as the graph size increases and that it 
provides accurate estimates. We also provide a comparison between the method by Yoshida \cite{yoshida2014almost}, and our sampling method.

\spara{Applications.}   Our scalable algorithm enables us to study some interesting characteristics of the central nodes. In particular, if $S$ is a set of nodes with high \BW, we focus on following questions.
\begin{itemize}

	\item[(1)] 
	How does the centrality of the most central set of nodes change in time-evolving networks? We study the {\sc DBLP} and the {\sc Autonomous} systems graphs. We mine interesting growth patterns, and 
	we compare our results to stochastic Kronecker graphs \cite{leskovec2005realistic}, a popular random graph model that mimics certain aspects of real-world networks. We observe that the Kronecker graphs behave similarly to real-world networks. 
	
	\smallskip
 
	\item[(2)]  Influence maximization has received a lot of attention since the seminal work of Kempe \emph{et al.} \cite{kempe2003maximizing}. 
Using our scalable algorithm we can compute a set of central nodes
that can be used as seeds for influence maximization. We find that
betweenness centrality is performing  relatively well compared to 
a state-of-the-art influence maximization algorithm.
	
	\smallskip
	 
	\item[(3)]  We study four  strategies for attacking a network 
	using  four centrality measures: betweenness, coverage, $\kappa$-path, and triangle  centrality. 
	Interestingly, we find that the  $\kappa$-path and triangle centralities
	can be more effective at destroying the connectivity of a graph.

\end{itemize}

%% file: tex/Related.tex
\section{Related Work} 
\label{sec:related}
 
\spara{Centrality measures.}  There exists a wide variety of centrality measures: degree centrality,  Pagerank \cite{page1999pagerank}, HITS \cite{kleinberg1999authoritative}, Salsa \cite{lempel2001salsa}, closeness centrality \cite{bavelas},  harmonic centrality \cite{boldi2014axioms}, betweenness centrality \cite{freeman}, random walk betweenness centrality \cite{newman2005measure}, coverage centrality \cite{yoshida2014almost}, $\kappa$-path centrality  \cite{alahakoon2011k}, Katz centrality \cite{katz}, rumor centrality \cite{shah2012rumor} are some of the important centrality measures. Boldi and Vigna proposed an axiomatic study of centrality measures \cite{boldi2014axioms}. In general, choosing a good centrality measure is application dependent~\cite{teng}. In the following we discuss in further detail  the centrality measure of our focus, the betweenness centrality.

\spara{Betweenness centrality (\BW)} is  a fundamental measure in  network analysis. The betweenness centrality index is attributed to Freeman \cite{freeman}.  
\BW has been used in a wide variety of graph mining applications. For instance, Girvan and Newman use \BW to find communities in social and biological networks \cite{girvan2002community}. In a similar spirit, Iyer \emph{et al}. use \BW to attack the connectivity of networks by iteratively removing the most central vertices \cite{iyer2013attack}.   

The fastest known exact algorithm for computing \BW exactly requires 
$O(mn)$ time in unweighted, and $O(nm+n^2 \log{m} )$ for weighted graphs~\cite{brandes2001faster,erdHos2015divide, sariyuce2012shattering}. There exist randomized algorithms \cite{bader2007approximating,pich,riondato2014fast} 
which provide either additive error or multiplicative error guarantees 
with high probability.

For \CMP, the state-of-the-art algorithm~\cite{yoshida2014almost} (and the only scalable proposed algorithm based on sampling) provides a mixed error guarantee, combining additive and multiplicative error. Specifically, this algorithm provides a solution whose centrality is at least $(1-\frac{1}{e}) \opt_k - \epsilon n^2$, by sampling $O({\log{n}}/{\epsilon^2})$ \emph{{\he}s}, where each {\he} is a set of all nodes on any shortest path between two random nodes with some assigned weights. \hide{ We show that even if $\opt_k = \Theta(n^2)$, for any constant $c>0$, sampling $o(n/\epsilon^2)$ number of {\he}s can result in an output  $S$ such that $B(S) < c \cdotp \opt_k$. We also improve this approximation guarantee to $(1-1/e-\epsilon)\opt_k$ by sampling random shortest paths.}

As we mentioned before, \CMP is APX-complete, and the best algorithm (i.e. classic greedy algorithm for maximizing the monotone-submodular functions)  using \emph{exact} computations of \BW provides $(1-1/e)$-approximation~\cite{fink2011maximum}. We call this greedy algorithm by  \exname algorithm, and works as follows: It starts with an empty set $S$. Then, at  any round, it selects a node $u$ that maximizes the \emph{adaptive betweenness centrality} (\ABW) of $u$ according to $S$ defined as
$$B(u|S) = \sum_{(s,t), s,t\neq u}\frac{\sigma_{s,t}(u|S)}{\sigma_{s,t}},$$
 where $\sigma_{st}(u|S)$ is the number of $s$-$t$ shortest paths that 
 do not pass through any node in $S$ and have $u$ as an internal node. The algorithm adds $u$ to $S$ and stops when $|S| = k$. Note that the \ABW is intimately connected to the \BW through the following well-known formula  \cite{yoshida2014almost}:
$$B(S \cup \{u\})  = B(S) + B(u|S).$$

%% file: tex/Proposed.tex
\section{Proposed Method}\label{sec:proposed}

In this section we provide our algorithm, \algoname (Hyper-EDge GrEedy), and a general framework for its analysis. We start by defining a \emph{{\he} sampler}, that will be used in \algoname.  

\begin{definition}[{\He} sampler]
We say that an algorithm $\A$ is a \textbf{{\he} sampler} for a function $C:2^{V} \rightarrow \mathbb{R}$ if it outputs a randomly generated subset of nodes $h\subseteq V$ such that
$$\forall S \subseteq V: \quad \text{Pr}_{h \sim\A}(h\cap S \neq \emptyset) = \frac{1}{\alpha}C(S),$$
where $\alpha$ is a normalizing factor, and independent of the set $S$. We call each $h$ (sampled by $\A$) a random {\he}, or in short, a \textbf{\he}. In this case, we say $C$ \textbf{admits} a {\he} sampler.  
\end{definition}

Our proposed algorithm \algoname  assumes the existence of a 
{\he} sampler and uses it in a black-box manner. Namely, \algoname  is oblivious to the specific mechanics of the {\he} sampler. 
The following lemma provides a simple {\he} sampler for \BW. 

\begin{lemma}\label{lem:BWsampler}
	The \BW admits a {\he} sampler.	
\end{lemma}

\begin{proof}
	Let $\A$ be an algorithm that selects two nodes $s,t \in V$ uniformly at random, selects a $s$-$t$ shortest path $P$ uniformly at random (this can be done in linear time $O(m+n)$ using bread-first-search from $s$, counting the number of shortest paths from $s$ and backward pointers; e.g. see \cite{riondato2014fast}), and finally outputs the internal nodes of $P$ (i.e., the nodes of $P$ except $s$ and $t$). 
	
	Now, suppose $h$ is an output of $\A$. Since the probability of choosing each pair is $\frac{1}{n(n-1)}$, and for a given pair $s,t$ the probability of $S\cap h \neq \emptyset$ is $\frac{\sigma_{s,t}(S)}{\sigma_{s,t}}$, for every $S\subseteq V$ we have
	\small
	$$\text{Pr}_{h \sim\A}(h\cap S \neq \emptyset) = \sum_{s,t\in V}\frac{1}{n(n-1)}\frac{\sigma_{s,t}(S)}{\sigma_{s,t}} = \frac{1}{n(n-1)}B(S).$$
	\normalsize
Also note that in this case, the normalizing factor is $\alpha = n(n-1) = \Theta(n^2)$. \qed
\end{proof}

For a subset of nodes $S\subseteq V$, and a set $\H$ of independently generated hyper-edges, define
 
$$\deg_{\H}(S) = |\set{h\in \H \mid h\cap S \neq \emptyset}|.$$ 

\noindent The pseudocode of our proposed algorithm \algoname is given in Algorithm~\ref{alg:hedge}.  First, it samples $q$ hyper-edges using the {\he} sampler $\A$  and then it runs a natural greedy procedure on $\H$.

\begin{algorithm}[!htp] 
\small
\BlankLine
{\bf Input:} A {\he} sampler $\A$ for \BW, number of hyper-edges $q$, and the size of the output set $k$.\\ 
{\bf Output:} A subset of nodes, $S$ of size $k$.
\BlankLine
\Begin{
	$\H \leftarrow \emptyset$\;
	\For{$i\in [q]$}{
		$h \sim \A$ (sample a random hyper-edge)\;
		$\H \leftarrow \H \cup \set{h}$\;
	}
	$S \leftarrow \emptyset$ \; 
	\While{$|S| <  k$}{
	    $u \leftarrow \argmax_{v\in V} \deg_{\H}(\set{v})$\;
	    $S \leftarrow S \cup \set{u}$\;
	    \For{$h\in \H$ such that $u \in h$}{
		$\H \leftarrow \H\setminus\set{h}$\;
	    }
	}
	\Return $S$\;
}
\caption{\label{alg:hedge}\algoname}
\end{algorithm}

\subsection{Analysis}\label{sec:analysis}
In this section we provide our general analytical framework for \algoname, which works with any {\he} sampler.  To start, define $B_{\H}(S) = \frac{\alpha}{|\H|} \deg_{\H}(S)$ as the  centrality (\BW) of a set $S$ according to the sample $\H$ of {\he}s, and for a graph $G$ let 
$$q(G,\epsilon) = \frac{3\alpha(\ell+k)\log(n)}{\epsilon^2 \opt_k},$$ 
where $n$ is the number of nodes in  $G$, and $\ell$ is a positive integer.   
\begin{lemma}\label{lem:num-samples}
Let $\H$ be a sample of independent {\he}s such that $|\H| \geq q(G, \epsilon)$. Then, for all $S\subseteq V$ where $|S| \leq k$ we have
 $\pr{|B_{\H}(S) - B(S)| \geq \epsilon\cdot\opt_k} < n^{-\ell}$.
\end{lemma}

 \begin{proof}
Suppose $S\subseteq V$ and $|S|\leq k$, and let $X_i$ be a binary random variable that indicates whether the $i$-th {\he} in $\H$ intersects with $S$. Notice that $deg_{\H}(S) = \sum_{i=1}^{|\H|} X_i$ and by the linearity of expectation $\E{deg_{\H}(S)} = |\H| \cdot \E{X_1}  = \frac{q}{\alpha}B(S)$. Using the independence assumption and the Chernoff bound, we obtain:
\begin{align*}
&\pr{|B_{\H}(S) - B(S)| \geq \delta\cdot B(S)} &=  \\
&\pr{\left| \frac{q}{\alpha}B_\H(S) - \frac{q}{\alpha}B(S) \right| \geq \frac{\delta q}{\alpha} \cdot B(S)} &= \\
 &\pr{|deg_{\H}(S) - \E{deg_{\mathcal{H}}(S)}| \geq \delta\cdot \E{deg_{\mathcal{H}'}(S)}} &\leq   \\
&2\exp\paran{-\frac{\delta^2}{3} \frac{q}{\alpha} B(S)}. 
\end{align*}
Now, by letting $\delta = \frac{\epsilon\opt_k}{B(S)}$ and substituting  the lower bound for $q(G, \epsilon)$  we obtain 
\begin{align*}
 \pr{|B_{\mathcal{H}} (S) - B(S)| \geq \epsilon \opt_k} 
&\leq n^{-(\ell+k)},
\end{align*}
and by taking a union bound over all possible subsets  $S\subseteq V$ of size  $k$ we obtain $|B_{\H}(S) - B(S)| < \epsilon\cdot\opt_k$ with probability at least $1-1/n^\ell$, for all such  subsets $S$. \qed
\end{proof}

Now, the following theorem shows that if the number of samples, i.e. $|\H|$, is at least $q(G,\epsilon/2)$, then \algoname provides a $(1-1/e-\epsilon)$-approximate solution.

\begin{theorem}\label{thm:approximation-rate}
If $\H$ is a sample of at least $q(G, \epsilon/2)$ {\he}s for some $\epsilon > 0$, and $S$ is the output of \algoname, we have $B(S)\geq (1-1/e-\epsilon)\opt_k$, with high probability.
\end{theorem}

\begin{proof}
Note that $B$ is 
(i) monotone since if $S_1\subseteq S_2$ then $B(S_1) \leq B(S_2)$, and (ii)  submodular since if $S_1 \subseteq S_2$ and $u \in V\setminus S_2$ then 
$B(S_2\cup\set{u}) - B(S_2) \leq B(S_1\cup\set{u}) - B(S_1)$. 

Similarly, $B_{\H}$ is monotone and submodular. Therefore, using the greedy algorithm (second part of \algoname) we have (see~\cite{nemhauser1978best})
$$B_\H(S) \geq  (1-1/e) B_\H(S') \geq (1-1/e)B_\H(S^*),$$
where 
$$S' = \argmax\limits_{T: |T| \leq k} B_\H(T), \quad \text{and}\quad S^* = \argmax_{T: |T| \leq k} B(T).$$
  Notice that $\opt_k = B(S^*)$. Since $|\H| \geq q(G, \epsilon/2)$, by Lemma~\ref{lem:num-samples} with probability $1-\frac{1}{n^\ell}$ we have

\small
\begin{align*}
B(S) &\geq B_\H(S) - \frac{\epsilon}{2}\opt_k \geq \paran{1-\frac{1}{e}}B_\H(S^*) - \frac{\epsilon}{2}\opt_k \\ 
&\geq \paran{1-\frac{1}{e}}\paran{B(S^*) - \frac{\epsilon}{2}\opt_k} - \frac{\epsilon}{2}\opt_k \geq \paran{1-\frac{1}{e}-\epsilon}\opt_k,
\end{align*} \normalsize
 
\noindent where we used the fact $B(S^*) = \opt_k$, and the proof is complete. \qed
\end{proof}

The total running time of \algoname depends on the running time of the {\he}  sampler and the greedy procedure. Specifically,  the total running time is 
$ O(t_{he} \cdot |\mathcal{H}|   +(n\log(n)+|\mathcal{H}|))$, where 
$t_{he}$ is the \emph{expected} required amount  of time for the sampler to output a single \he.   The first term corresponds to the total required time for sampling, and the second term corresponds to an almost linear time implementation of greedy procedure as in \cite{borgs2014maximizing}.

\begin{remark}
Note that if $\opt_k = \Theta(n^2)$, the sample complexity in Theorem~\ref{thm:approximation-rate} becomes $O\paran{\frac{k\log(n)}{\epsilon^2}}$. We provide the first theoretical study on this assumption in Sect.~\ref{sec:optn2}.
\end{remark}

Finally, we provide a lower bound on the sample complexity of \algoname, in order to output a $(1-1/e-\epsilon)$-approximate solution. This lower bound is still valid even if $\opt_k = \Theta(n^2)$.

\begin{theorem}
\label{thm:lowerbound}
	In order to output  a set $S$ of size $k$ such that $B(S) \geq (1-1/e-\epsilon)\opt_k$ w.h.p.,  the sample size in both \algoname and \cite{yoshida2014almost}'s algorithm needs to be at least $\Omega \paran{\frac{n}{\epsilon^2}}$.
\end{theorem}

\begin{proof}
Define a graph $A=(V_A,E_A)$ where 
$$V_{A} = \set{(i,j)\mid 1\leq i \leq \epsilon \sqrt{n} \text{ and } 1 \leq j \leq \sqrt{n}},$$
 and two nodes $(i,j)$ and $(i',j')$ are connected if $i=i'$ or $j=j'$. Note that the distance between every pair of nodes is at most 2, and there are at most 2 shortest paths between a pair of nodes in $A$. Let $G$ be a graph of size $n$ which has $(1-\epsilon)n$  isolated nodes, and $A$ as its largest connected component.  
 
\underline{Claim:} 	If $k = \epsilon n$, then $\opt_k = \sqrt{n}(\sqrt{n}-1)\cdot \epsilon\sqrt{n}(\epsilon\sqrt{n}-1)= \Theta(\epsilon^2 n^2) = \Theta(n^2)$, since $\epsilon$ is a constant.

\noindent	Obviously, all the isolated nodes in $G$ have zero \BW. So, the optimal set, $S^*$, is $A$ (which is already of size $k=\epsilon n$). Now, if for two nodes $s,t$ in $G$, there is a shortest path with an internal node in $A$, we have $s,t \in V_A$ such that $s=(i,j)$ and $t=(i',j')$ where $i\neq i'$ and $j\neq j'$. In this case, there are exactly two $s$-$t$ shortest paths with exactly one internal node. Finally, the number of such pairs is exactly $\sqrt{n}(\sqrt{n}-1)\cdot \epsilon\sqrt{n}(\epsilon\sqrt{n}-1)$.

Now note that in both \algoname and the algorithm of \cite{yoshida2014almost}, we first choose a pair of nodes in $s, t$ in $G$, and if $s$ and $t$ are not in the same connected component, the returned {\he} is an empty set. Therefore, in order to have a non-empty {\he} both nodes should be chosen from $V_A$, which occurs with probability $\epsilon^2$. Thus, sampling $o(n/\epsilon^2)$ {\he} results in reaching to at most $o(n) = o(|A|) = o(k)$ nodes, and so, the algorithm will not be able to tell the difference between the isolated nodes and many (arbitrarily large) number of nodes in $A$ as they were not detected by any {\he}. This concludes our proof. \qed
\end{proof}

\begin{remark}
	Theorem~\ref{thm:lowerbound} implies that the approximation guarantee of the Algorithm proposed in~\cite{yoshida2014almost} that uses   $M=O(\log(n)/\epsilon^2)$ samples is not correct.
\end{remark}

\subsection{Beyond betweenness centrality}\label{sec:beyond}
Suppose $C:2^{V}\rightarrow \mathbb{R}^{\geq 0}$ is a centrality measure that is also defined on subset of nodes. Clearly, if $C$ is monotone-submodular and admits a {\he} sampler, the algorithm \algoname can be applied to and all the results in this section hold for $C$. Here, we give a couple of examples of such centrality measures, and it is easy to verify their monotonicity and submodularity.

\spara{Coverage centrality.} The coverage centrality~\cite{yoshida2014almost} for a set $S\subseteq V$ is defined as 
$C(S) = \sum_{(s,t)\in V^2} P_{s,t}(S),$
where $P_{s,t}(S)$ is 1 if $S$ has an internal node on any $s$-$t$ shortest path, and 0 otherwise.
The coverage centrality admits a {\he} sampler $\A$ as follows: uniformly at random pick two nodes $s$ and $t$. By running a breadth-first-search from $s$, we output every node that is on at least one shortest path from $s$ to $t$. Note that for every subset of nodes $\text{Pr}_{h \sim\A}(h\cap S \neq \emptyset) = \frac{1}{n(n-1)}C(S)$.

\spara{$\kappa$-Path centrality} Alahakoon et al. introduced the  $\kappa$-path centrality of a node~\cite{alahakoon2011k}.Their notion naturally generalizes to any subset of nodes as 
$C(S) = \sum_{s\in V} P_\kappa^s(S),$
where $P_\kappa^s(S)$ is the probability that a \emph{random simple path} of length $\kappa$ starting at $s$ will pass a node in $S$: a random simple path starting at node $s$ is generated by running a random walk that always chooses an unvisited neighbor uniformly at random, and stops after $\kappa$ of edges being traversed or if there is no unvisited neighbor. Note that $\kappa$-path centrality is a generalization of \emph{degree centrality} by letting $\kappa=1$ and considering sets of single nodes.

Obviously, $\kappa$-path centrality admits a {\he} sampler based on its definition: Let $\A$ be an algorithm that picks a node uniformly at random, and generates a random simple path of length at most $\kappa$, and outputs the generated simple path as a {\he}. Therefore, for any subset $S$ we have
$\text{Pr}_{h\sim \A}\paran{h \cap S \neq \emptyset} = \frac{1}{n} \sum_{s\in V} P_\kappa^s(S) = \frac{1}{n} C(S)$.

%% file: tex/assumption.tex
\section{On the $\opt_k=\Theta(n^2)$ Assumption}\label{sec:optn2}

 Recall that all additive approximation guarantees for \BW as well as all existing approximation guarantees for \ABW  involve an error term which grows as $\Theta(n^2)$.
In this Section we provide strong theoretical evidence in favor 
of the following question:  ``Why does prior work which relies heavily on the strong assumption  
that $\opt_k=\Theta(n^2)$ perform well on real-world networks?'' 
We quote Yoshida \cite{yoshida2014almost}: 
{\em This additive error should not be critical in most applications, as numerous real-world graphs have vertices of centrality $\Theta(n^2)$.}  

In general, this strong assumption is not true. The complete graph is perhaps the trivial counterexample to this assumption as the centrality of any set of nodes is zero, and  thus $\opt_k \ll \Theta(n^2)$.  Another more interesting example is the hypercube.
 
\spara{Hypercube.} The hypercube $Q_r$ (with $n=2^r$ nodes) is a \emph{vertex-transitive graph}, i.e., for each pair of nodes there is an automorphism that maps one onto the other~\cite{west2001introduction}. So, the centrality of the nodes are the same.

First consider the \textbf{coverage centrality}, and lets count how many pairs of nodes have a shortest path that pass the node $(0,\ldots, 0)$. Note that  $a, b \in \set{0,1}^r$ have a shortest path that passes the origin if and only if $\forall i\in {1,\ldots, r}: a_i b_i = 0$. To count the number of such pairs, we have to first choose a subset $I\subseteq \set{1,\ldots, r}$ of the bits that are non-zero either in $a$ or $b$, in $\binom{r}{|I|} $ ways, and partition the bits of $I$ between $a$ and $b$ (in $2^{|I|}$ ways). Therefore, the number of $(a,b) \in V^2$  pairs that their shortest path passes the node $(0,\ldots,0)$ is
$$\sum_{i=0}^r {r \choose i} 2^i = (1+2)^r =3^{\log(n)} = n^{\log(3)} = o(n^2).$$
So, the maximum coverage centrality of a node is at most $n^{\log(3)}$ (since we counted the endpoints as well, but should not have). Now by submodularity of the coverage centrality we have $\opt_k \leq k n^{\log(3)} = O(n^{1+\log(3)}) = o(n^2)$.
Finally, since the betweenness centrality is no more than the coverage centrality, we have the similar result for betweenness centrality as well.

We show that  the $\opt_k = \Theta(n^2)$ assumption holds for two important classes of graphs: graphs of bounded treewidth networks, and for certain stochastic graph models that generate scale-free and small-world networks, known as \emph{random Apollonian networks} 
\cite{frieze2014some}.

\subsection{Bounded treewidth graphs} 
We start by defining the notion of the treewidth of an undirected graph.

\begin{definition}[Treewidth]
For an undirected graph $G=(V,E)$, a \textbf{tree decomposition} is a tree $T$ with nodes $V_1,\ldots, V_r$ where each $V_i$ is (assigned to) a subset of $V$ such that (i) for every $u\in V$ there exists at least an $i$ where $u\in V_i$, (ii) if $V_i$ and $V_j$ both contains a node $u$, then $u$ belongs to every $V_k$ on the unique shortest path from $V_i$ to $V_j$ in $T$, and (iii) for every edge $(u,v)\in E$ there exists a $V_i$ such that $u,v \in V_i$. The \textbf{width} of the tree decomposition $T$ is defined as $max_{1\leq i\leq r} |V_i|-1$, and the \textbf{treewidth} of the graph $G$ is the minimum possible width of any tree decomposition of $G$.  	
\end{definition}
Now, we have the following theorem.

\begin{theorem}\label{thm:boundedtreewidth}
Let $G=(V,E)$ be an undirected, connected graph of bounded treewidth. Then $\opt_k = \Theta(n^2)$.
\end{theorem} 

\begin{proof} 
Suppose $w$ is the treewidth of $G$, which is a constant (bounded). It is known that any graph of treewidth $w$ has a balanced vertex separator\footnote{Means a set of nodes $\Gamma$ such that $V\setminus \Gamma = A \cup B$, where $A$ and $B$ are disjoint and both have size $\Theta(n)$.\hide{ Note that $A$ and $B$ do not need to be connected subgraphs.}}
 $S \subseteq V$ of size at most $w+1$~\cite{robertson1986graph}. This implies that $O(n^2)$ shortest paths pass through $S$. Since $|S|=w+1=\Theta(1)$, there exists at least one vertex $u \in S$ such that $B(u)=\Theta(n^2)$. Hence, $\opt_1=\Theta(n^2)$, and since $\opt_1 \leq \opt_k$ we have $\opt_k=\Theta(n^2)$.
\end{proof}  

It is worth emphasizing that the 
classical Barab{\'a}si-Albert random tree model \cite{barabasi1999emergence,mori2005maximum} belongs to this category. For a recent study of the treewidth parameter on real-world networks, see \cite{adcock2014tree}. 

 

\subsection{Scale-free, small-world networks} We show that $\opt_k=\Theta(n^2)$ for random Apollonian networks. 
 Our proof for the latter model relies on a technique developed by Frieze and Tsourakakis \cite{frieze2014some} and carries over for random unordered increasing $k$-trees \cite{gao2009degree}. A \emph{random Apollonian network} (RAN) is a network that is generated iteratively. The RAN generator takes as input  the desired number of nodes $n\geq 3$ and runs as follows: 
\squishlist
 \item Let $G_3$ be the triangle graph, whose nodes are $\set{1,2,3}$, and drawn in the plane. 
 \item {\bf for} $t \leftarrow 4$ to $n$: 
\squishlist
  \item  Sample a face $F_t=(i,j,k)$ of the planar  graph $G_{t-1}$ uniformly at random, except for the outer face. 
  \item  Insert the new node $t$ inside this face connecting it to $i,j,k$. 
\squishend
\squishend
 
Figure~\ref{fig:apollonian}(a) shows an instance of a RAN for $n=100$. 
The triangle is originally 
 embedded on the plane  as an equilateral triangle. Also, 
 when a new node $t$ chooses its face $(i,j,k)$ it is embedded in the barycenter of the corresponding triangle and connects to $i,j,k$ 
via the  straight lines: $(i,t), (j,t)$, and $(k,t)$. It has been shown that the diameter of
a RAN is $O(\log(n))$ with high probability \cite{frieze2014some,ebrahimzadeh2013longest}.

\begin{figure}[t]
\centering
\begin{tabular}{@{}c@{}@{\ }c@{}}
 \includegraphics[width=0.33\textwidth]{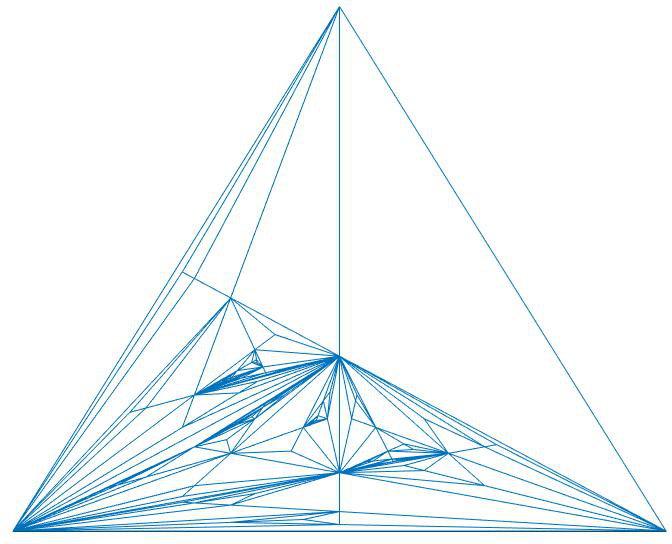} &
 \includegraphics[width=0.63\textwidth]{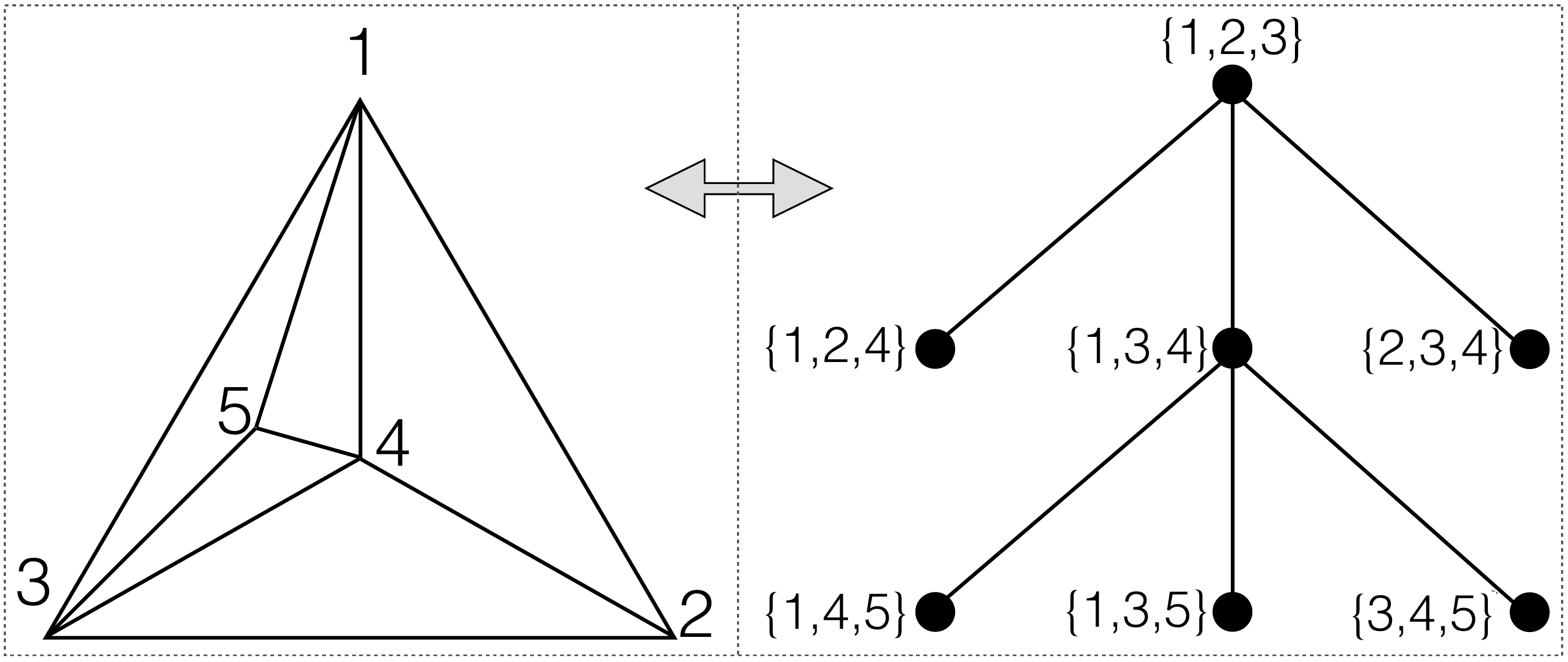}\\
(a) & (b) \\ 
\end{tabular}
\caption{\label{fig:apollonian} (a)    An instance of a random Apollonian network for $n=100$. (b) Bijection between RANs and random ternary trees.}
\vspace{-4mm}
\end{figure}

At any step, we refer to set of candidate faces as the set of active faces. 
Note that there is a bijection between  the active faces  of a RAN 
and the leaves of random ternary trees  as illustrated in Figure~\ref{fig:apollonian}(b), and noticed first by \cite{frieze2014some}.  

We shall make use of the following formulae for the number of nodes ($v_t$), edges ($e_t$) and faces ($f_t$; excluding the outer face) after  $t$ steps in a RAN $G_t$: 
$$v_t=t,~~ e_t=3t-6,~~ f_t=2t-5.$$  
Note that  $f_t = \Theta(v_t)= \Theta(t)$.

%
%
%
%

\hide{
 \begin{figure*}[!ht]
\centering
\includegraphics[width=0.35\textwidth]{figures/bijection} 
\caption{\label{fig:bijection}  RANs as random ternary trees. Source~\cite{frieze2014some}. }
\end{figure*}
}

\hide{
\noindent Before we present the proof in detail it is worth sketching its main idea: based on the bijection between RANs and random ternary trees we are able to find a vertex in the tree that  }

\begin{theorem} 
Let $G$ be a  RAN of size $n$. Then $\opt_k=\Theta(n^2)$.
\end{theorem}

\begin{proof} 

Note that removing a node from the random ternary tree $T=(V_T,E_T)$ (as in Fig.~\ref{fig:apollonian}(b)) corresponds to removing three nodes from $G$, corresponding to a face $F$ that existed during the RAN generation process. Clearly, the set of these three nodes is a  vertex separator that separates the nodes inside and outside of $F$. Therefore, all  nodes in the tree except for the root $r$ correspond to a vertex separator in $G$. Observe that the leaves in $T$ correspond to the set of active faces, and thus, $T$ has $f_n=2n-5$ leaves after $n$ steps. 

We claim that there exists an edge $(F,F') \in E_T$ (recall that the nodes of $T$ are active faces during the generating process of $G$)  such that the removal of $e$ from $E_T$ results in two subtrees with $\Theta(n)$ leaves. We define $g: V_T \rightarrow \field{Z}$ to be the function that returns for a node $v \in V_T$ the number of leaves of the subtree rooted at $v$. Hence, $g(r) = 2n-5, g(u)=1$ for any leaf $u$.  
To find such an edge consider the following algorithm. We start from the root $r$ of $T$ descending to the bottom according to the following rule which creates a sequence of nodes $u_0=r,u_1,u_2,\ldots$: we set $u_{i+1}$ to be the child with most leaves among the tree subtrees rooted at the three children of $u_i$. We stop when we first find a node $u_i$ such that $g(u_i) \geq cn$ and $g(u_{i+1}) < cn$ for some constant $c$. Clearly, $g(u_{i+1}) \geq cn/3=\Theta(n)$, by pigeonhole principle. So, let $F = u_i$ and $F' = u_{i+1}$. 

Now suppose $F' = \set{x,y,z}$, and consider removing $x,y$, and $z$ from $G$. Clearly, $F' \neq r$ as $F'$ is a child of $F$.  Also, due to the construction of a RAN, after removing $x,y,z$, there are exactly two connected components, $G_1$ and $G_2$. Also, since $cn/3\leq g(F')<cn$, the number of nodes in each of $G_1$ and $G_2$ is $\Theta(n)$.

\hide{ To see why notice that since the number of leaves in each of the tree components $T_1,T_2$ of $T\backslash (s,t)$ is $\Theta(n)$, there exist at least $\Theta(n/3)$ ancestors which further implies that there exists $\Theta(n/9)=\Theta(n)$ vertices in $G_1,G_2$.}

\hide{
, we claim that the two resulting subtrees  induced t

have $\Theta(n)$ leaves each, implies that the corresponding graph components have  $\Theta(n)$ vertices each. Notice that this claim is non-trivial as in principle a subtree can span $\Theta(n)$ leaves and have diameter $\Theta(n)$. However, our claim is directly implied by the fact that the every vertex that is not a leaf has  3 children and the diameter of $T$ is $\Theta(\log{n})$ .  }

Finally, observe that at least one of the three nodes $x,y,z$ must  have betweenness centrality score $\Theta(n^2)$, as the size of the separator is 3 and there exist $\Theta(n^2)$ paths that pass through it (connecting the nodes in $G_1$ and in $G_2$). Therefore, $\opt_1 \geq \max\set{B(x),B(y),B(z)} = \Theta(n^2)$, and since $\opt_1 \leq \opt_k$ we have $\opt_k = \Theta(n^2)$.
\end{proof} 

We believe that this type of coupling can be used to prove similar results for other stochastic graph models.

%% file: tex/Experiments.tex
\section{Experimental Results}
\label{sec:exp} 
 
 In this section we present our experimental results. We first start by comparing \algoname (our sampling based algorithm) with  \exname (the exhaustive algorithm defined in Sect.~\ref{sec:related}) and show that the centrality of \algoname's output is close to the centrality of \exname's output, with great speed-up. This part is done for 3 small graphs as \exname cannot scale to larger graphs. 
 
We then compare our sampling method with the method presented in ~\cite{yoshida2014almost}. We show that, although our method stores less per each {\he}, it does not loose its accuracy.

Equipped with our scalable algorithm, \algoname, we will be able to focus on some of the interesting characteristics of the central nodes: (i) How does their centrality change over time in evolving graphs? (ii) How influential are they? and (iii) How does the size of the largest connected component change after removing them?

In our experiments, we assume the graphs are simple (no self-loop or parallel edge) but the edges can be directed. We used publicly available datasets in our experiments\footnote{\url{http://snap.stanford.edu} and \url{http://konect.uni-koblenz.de/networks/dblp_coauthor}}.
\algoname is implemented in {\sc C++}.

\input{tex/Efficiency.tex}

\input{tex/Applications.tex}

%% file: tex/Efficiency.tex


\subsection{Accuracy and time efficiency} 

Table~\ref{table:comparison} shows the results of \exname 
and \algoname on three graphs for which we were able to run 
\exname. The fact that \exname is able to run only on 
networks of this scale indicates already the value of \algoname's scalability. As we can see, \algoname results in significant speedups and 
negligible loss of accuracy. 


In Table~\ref{table:comparison} the centrality of the output sets and the speed up gained by \algoname is given, and as shown, \algoname gives a great speedup with almost the same quality (i.e., the centrality of the output) of \exname. The centrality of the outputs are \emph{scaled} by $\frac{1}{n(n-1)}$, where $n$ is the number of nodes in each graph.
Motivated by the result in Sect.~\ref{sec:optn2} we run \algoname using $k\log(n)/\epsilon^2$ {\he}s for $\epsilon=0.1$, and for  each case, ten times (averages are reported). For sake of comparison, these experiments  were executed using a single machine with  Intel Xeon {\sc cpu} at 
2.83GHz and with 36GB {\sc ram}.

\input{tables/speed.tex}

\subsection{Comparison against \cite{yoshida2014almost}}
\label{sec:compare}
 We compare our method against Yoshida's algorithm (\ya) \cite{yoshida2014almost} on four undirected graphs (as \ya runs on undirected graphs). We use Yoshida's implementation which he kindly provided to us.  Note that Yoshida's algorithm applies a different sampling method than ours: it is based on sampling random $s$-$t$ pairs of nodes and assigning weights to \textit{every} node that is on any $s$-$t$ shortest  path, whereas in our method we only pick one randomly chosen $s$-$t$ shortest path with no weight on the nodes.


%

\ya and \algoname use $\frac{2\log(2n^3)}{\epsilon^2}$ and $\frac{k\log(n)}{\epsilon^2}$ samples, respectively, where $n$ is the number of nodes in the graph, and we set $\epsilon=0.1$.
We also run a variation of our algorithm, $\algoname_=$,
which is essentially \algoname but with $\frac{2\log(2n^3)}{\epsilon^2}$ samples. This allows a more fair comparison between the methods. 


$ $

Table~\ref{table:compare-with-yoshida} shows the estimated centrality of the output sets, and the number of samples each algorithm uses. Surprisingly, \ya does not outperform 
 $\algoname_=$, despite the fact that it maintains extra information.
Finally, our proposed algorithm \algoname is consistently better than the other two algorithms.
%

\input{tables/compare.tex}

%% file: tables/speed.tex
\begin{table}[!htbp]
\centering
\small
\caption{\small \algoname vs. \exname: centralities and speedups.}
\label{table:comparison}

\begin{tabular}{lccc|c|c|r}
\cline{5-6}
                                                      &                                            &                                             &     & \multicolumn{2}{c|}{\textbf{Algorithms}} &                              \\ \hline
\multicolumn{1}{|l|}{\textbf{GRAPHS}}                 & \multicolumn{1}{l|}{\#nodes}               & \multicolumn{1}{l|}{\#edges}                & $k$ & \exname            & \algoname           & \multicolumn{1}{r|}{speedup} \\ \hline
\multicolumn{1}{|l|}{\multirow{3}{*}{\textbf{ca-GrQd}}}        & \multicolumn{1}{c|}{\multirow{3}{*}{5242}} & \multicolumn{1}{c|}{\multirow{3}{*}{14496}} & 10  & 0.242              & 0.241               & \multicolumn{1}{c|}{2.616}   \\ \cline{4-7} 
\multicolumn{1}{|l|}{}                                & \multicolumn{1}{l|}{}                      & \multicolumn{1}{l|}{}                       & 50  & 0.713              & 0.699               & \multicolumn{1}{c|}{2.516}   \\ \cline{4-7} 
\multicolumn{1}{|l|}{}                                & \multicolumn{1}{l|}{}                      & \multicolumn{1}{l|}{}                       & 100 & 0.974              & 0.951               & \multicolumn{1}{c|}{2.217}   \\ \hline
\multicolumn{1}{|l|}{\multirow{3}{*}{\textbf{p2p-Gnutella08}}} & \multicolumn{1}{c|}{\multirow{3}{*}{6301}} & \multicolumn{1}{c|}{\multirow{3}{*}{20777}} & 10  & 0.013              & 0.011               & \multicolumn{1}{c|}{6.773}   \\ \cline{4-7} 
\multicolumn{1}{|l|}{}                                & \multicolumn{1}{l|}{}                      & \multicolumn{1}{l|}{}                       & 50  & 0.036              & 0.035               & \multicolumn{1}{c|}{6.478}   \\ \cline{4-7} 
\multicolumn{1}{|l|}{}                                & \multicolumn{1}{l|}{}                      & \multicolumn{1}{l|}{}                       & 100 & 0.053              & 0.051               & \multicolumn{1}{c|}{6.117}   \\ \hline
\multicolumn{1}{|l|}{\multirow{3}{*}{\textbf{ca-HepTh}}}       & \multicolumn{1}{c|}{\multirow{3}{*}{9877}} & \multicolumn{1}{c|}{\multirow{3}{*}{25998}} & 10  & 0.165              & 0.164               & \multicolumn{1}{c|}{4.96}    \\ \cline{4-7} 
\multicolumn{1}{|l|}{}                                & \multicolumn{1}{l|}{}                      & \multicolumn{1}{l|}{}                       & 50  & 0.498              & 0.497               & \multicolumn{1}{c|}{4.729}   \\ \cline{4-7} 
\multicolumn{1}{|l|}{}                                & \multicolumn{1}{l|}{}                      & \multicolumn{1}{l|}{}                       & 100 & 0.747              & 0.745               & \multicolumn{1}{c|}{4.473}   \\ \hline
\end{tabular}

\end{table}

%% file: tables/compare.tex
\begin{table}[!htbp]
\centering
\small
\caption{Comparison against \ya}
\label{table:compare-with-yoshida}
\begin{tabular}{lc|l|l|l|c|l|l|}
\cline{3-8}
                                                             & \multicolumn{1}{l|}{}             & \multicolumn{3}{c|}{\textbf{Betw. Centrality}}   & \multicolumn{3}{c|}{\textbf{\# of Samples}}                           \\ \hline
\multicolumn{1}{|l|}{\textbf{GRAPHS}}                        & \multicolumn{1}{l|}{\textit{$k$}} & \textit{\texttt{Y-ALG}} & \textit{$\algoname_{=}$} & \textit{\algoname} & \multicolumn{1}{l|}{\textit{\texttt{Y-ALG}}} & \textit{$\algoname_{=}$} & \textit{\algoname} \\ \hline
\multicolumn{1}{|l|}{\multirow{3}{*}{\textbf{CA-GrQc}}}      & \textit{10}                       & 0.208            & 0.214          & 0.215        & \multicolumn{2}{c|}{\multirow{3}{*}{5278}}             & 8565         \\ \cline{2-5} \cline{8-8} 
\multicolumn{1}{|l|}{}                                       & \textit{50}                       & 0.484            & 0.483          & 0.49         & \multicolumn{2}{c|}{}                                  & 42822        \\ \cline{2-5} \cline{8-8} 
\multicolumn{1}{|l|}{}                                       & \textit{100}                      & 0.569            & 0.568          & 0.577        & \multicolumn{2}{c|}{}                                  & 85643        \\ \hline
\multicolumn{1}{|l|}{\multirow{3}{*}{\textbf{CA-HepTh}}}     & \textit{10}                       & 0.151            & 0.151          & 0.154        & \multicolumn{2}{c|}{\multirow{3}{*}{5658}}             & 9198         \\ \cline{2-5} \cline{8-8} 
\multicolumn{1}{|l|}{}                                       & \textit{50}                       & 0.403            & 0.4            & 0.409        & \multicolumn{2}{c|}{}                                  & 45989        \\ \cline{2-5} \cline{8-8} 
\multicolumn{1}{|l|}{}                                       & \textit{100}                      & 0.534            & 0.533          & 0.547        & \multicolumn{2}{c|}{}                                  & 91978        \\ \hline
\multicolumn{1}{|l|}{\multirow{3}{*}{\textbf{ego-Facebook}}} & \textit{10}                       & 0.924            & 0.932          & 0.933        & \multicolumn{2}{c|}{\multirow{3}{*}{5121}}             & 8304         \\ \cline{2-5} \cline{8-8} 
\multicolumn{1}{|l|}{}                                       & \textit{50}                       & 0.959            & 0.957          & 0.959        & \multicolumn{2}{c|}{}                                  & 41519        \\ \cline{2-5} \cline{8-8} 
\multicolumn{1}{|l|}{}                                       & \textit{100}                      & 0.962            & 0.96           & 0.964        & \multicolumn{2}{c|}{}                                  & 83038        \\ \hline
\multicolumn{1}{|l|}{\multirow{3}{*}{\textbf{email-Enron}}}  & \textit{10}                       & 0.329            & 0.335          & 0.335        & \multicolumn{2}{c|}{\multirow{3}{*}{6445}}             & 10511        \\ \cline{2-5} \cline{8-8} 
\multicolumn{1}{|l|}{}                                       & \textit{50}                       & 0.644            & 0.646          & 0.65         & \multicolumn{2}{c|}{}                                  & 52552        \\ \cline{2-5} \cline{8-8} 
\multicolumn{1}{|l|}{}                                       & \textit{100}                      & 0.754            & 0.756          & 0.762        & \multicolumn{2}{c|}{}                                  & 105104       \\ \hline
\end{tabular}
\end{table}

%% file: tex/Applications.tex
\subsection{Applications}
\label{sec:applications}

For the next three experiments, we consider three more larger graphs that \algoname can handle due to its scalability: {\sc email-Enron}, {\sc Brightkite}, and {\sc Epinions} networks. They  consist of   (36\,692, 183\,831), (58\,228,214\,078), and (75\,879, 508\,837)  nodes, and edges respectively. These experiments are based on orderings defined over the set of nodes as follows: we generate $100\log(n)/\epsilon^2$ hyper-edges, where $n$ is the number of nodes, and $\epsilon=0.25$. Then, we order the nodes based on the order \algoname picks the nodes.

For the sake of comparison, we ran \algoname using the coverage and $\kappa$-path (for $\kappa=2$) centralities, since both of them admit {\he} sampler as we showed in Sect.~\ref{sec:beyond}. Also, we considered a fourth centrality that we call triangle centrality, where the centrality of a set of nodes $S$ equals to the number of triangles that intersect with $S$. For the triangle centrality, we run \exname as computing this centrality is easy and scalable to large graphs\footnote{\exname for the triangle centrality, at every iteration simply chooses a node that is incident with more number of new triangles.}. All these experiments are run ten times, and we report the average values.
 

\subsection*{Time evolving networks}
The study of empirical properties of time-evolving real-world networks has attracted a lot of interest over the recent years, see for example
\cite{leskovec2007graph,mitzenmacher2015scalable}.
In this section we investigate how \BW of the most central nodes changes as a function of time.

We study two temporal datasets, the \textsf{DBLP} \footnote{Timestamps are in Unix time and can be negative.} and \textsf{Autonomous Systems} (AS) datasets. We also generate stochastic Kronecker graphs on $2^i$ vertices for $i\in\set{8,\ldots,20}$, using $\left(\begin{smallmatrix} 0.9 & 0.5\\ 0.5 & 0.2 \end{smallmatrix}\right)$ as the core-periphery seed matrix. We assume that the $i$-th time snapshot for Kronecker graphs corresponds to $2^{i}$ vertices, for $i=8,\ldots,20$. Note that in these evolving sets, the number of nodes also increases along with new edges. Also, note that the main difference between \textsf{DBLP} and \textsf{Autonomous Systems} is that for \textsf{DBLP} edges and nodes only can be added, where in \textsf{Autonomous Systems} nodes and edges can be increased and decreased.

The results are plotted in logarithmic scale (Fig.~\ref{fig:timeevolving}), and as shown, we observe that the centrality of the highly central set of nodes increases. Also, we observe that the model of stochastic Kronecker graphs behaves similar to the real-world evolving networks with respect to these parameters.

\begin{figure}[h]
\centering
\begin{tabular}{@{}c@{}@{\ }c@{}@{\ }c@{}}
\includegraphics[width=0.27\textwidth]{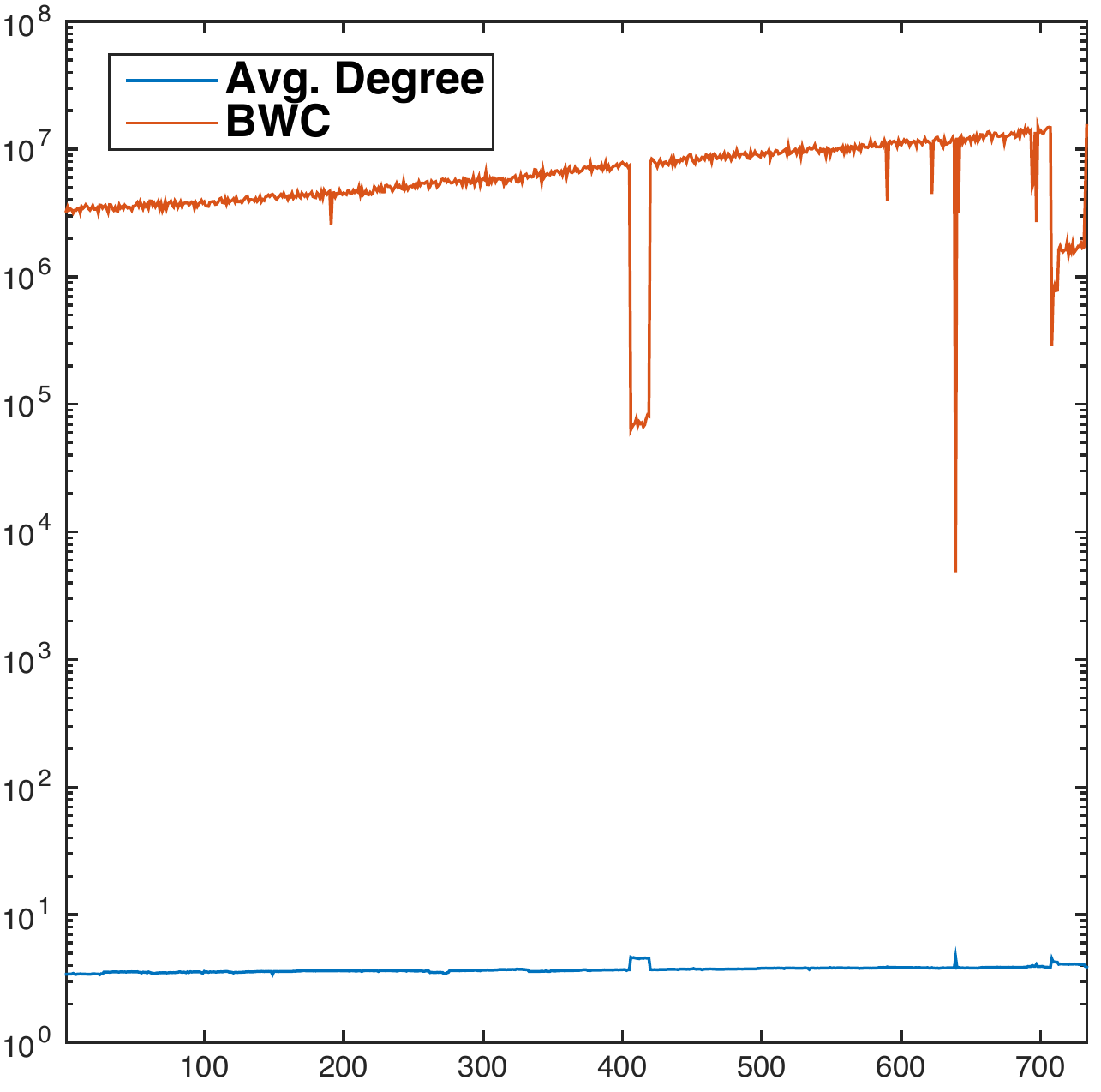} &
\includegraphics[width=0.26\textwidth]{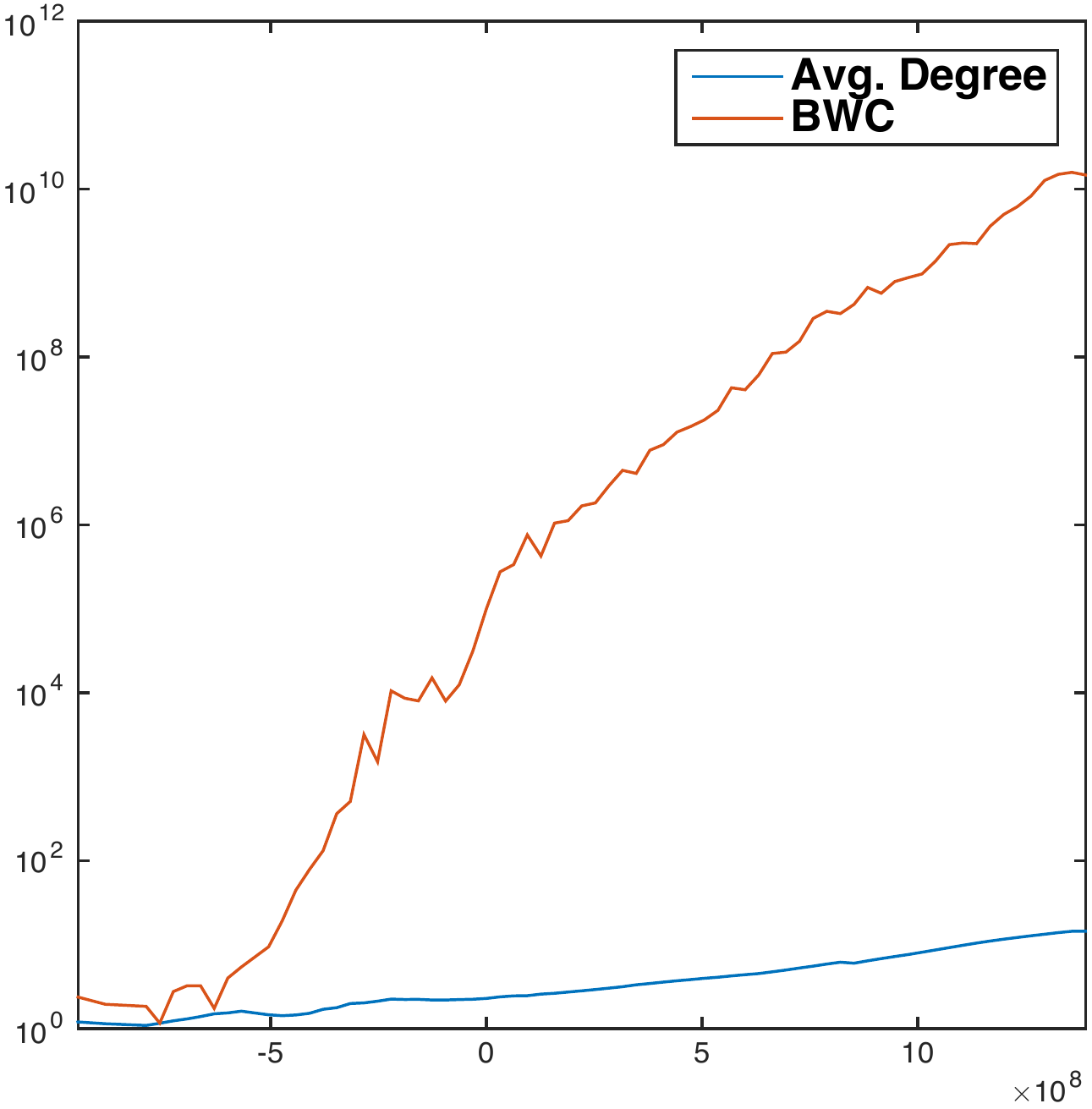} & 
\includegraphics[width=0.27\textwidth]{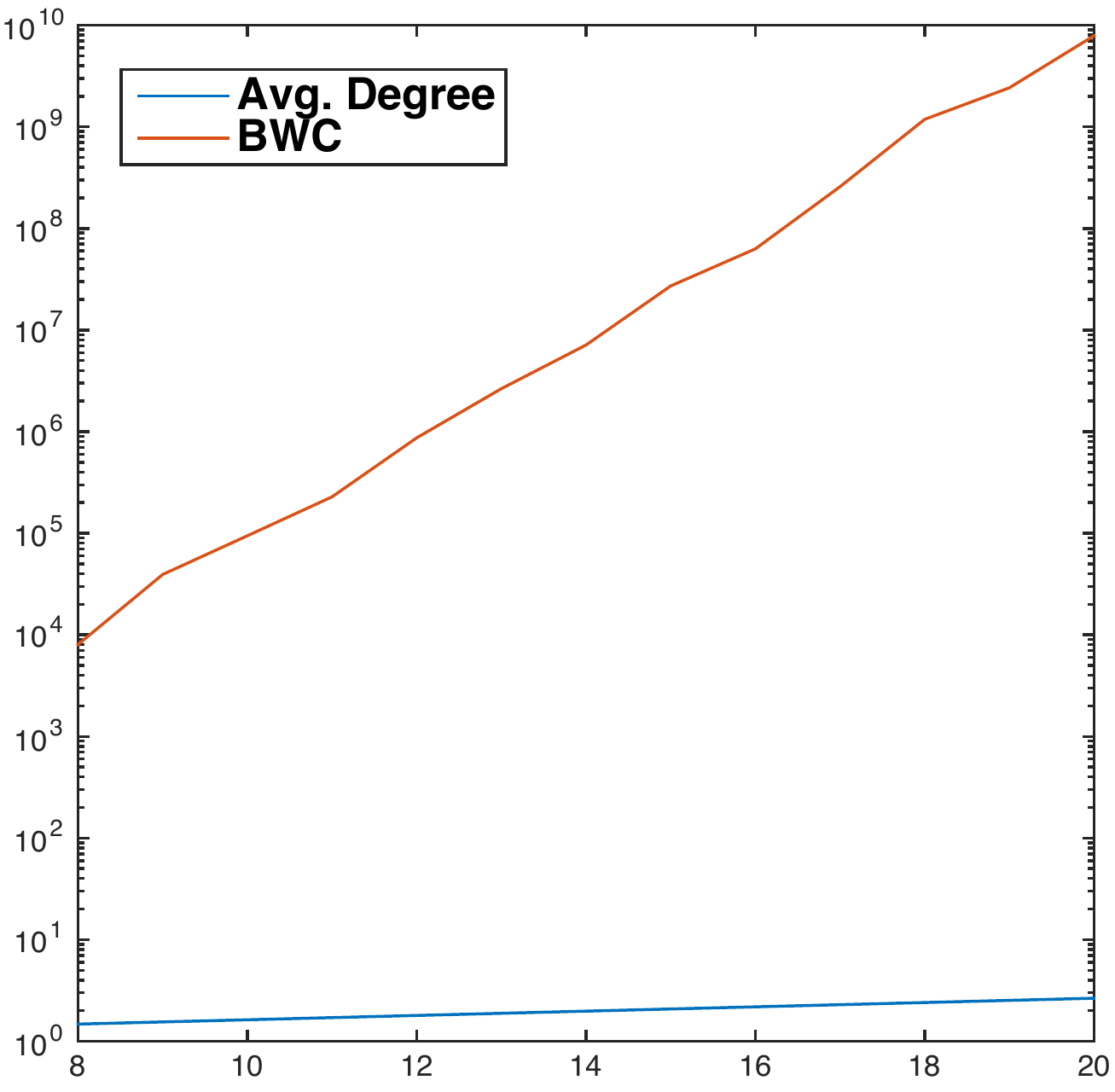} \\
(a) AS: $k=1$ & (b) DBLP: $k=1$ & (c) KG: $k=1$\\  
 \includegraphics[width=0.27\textwidth]{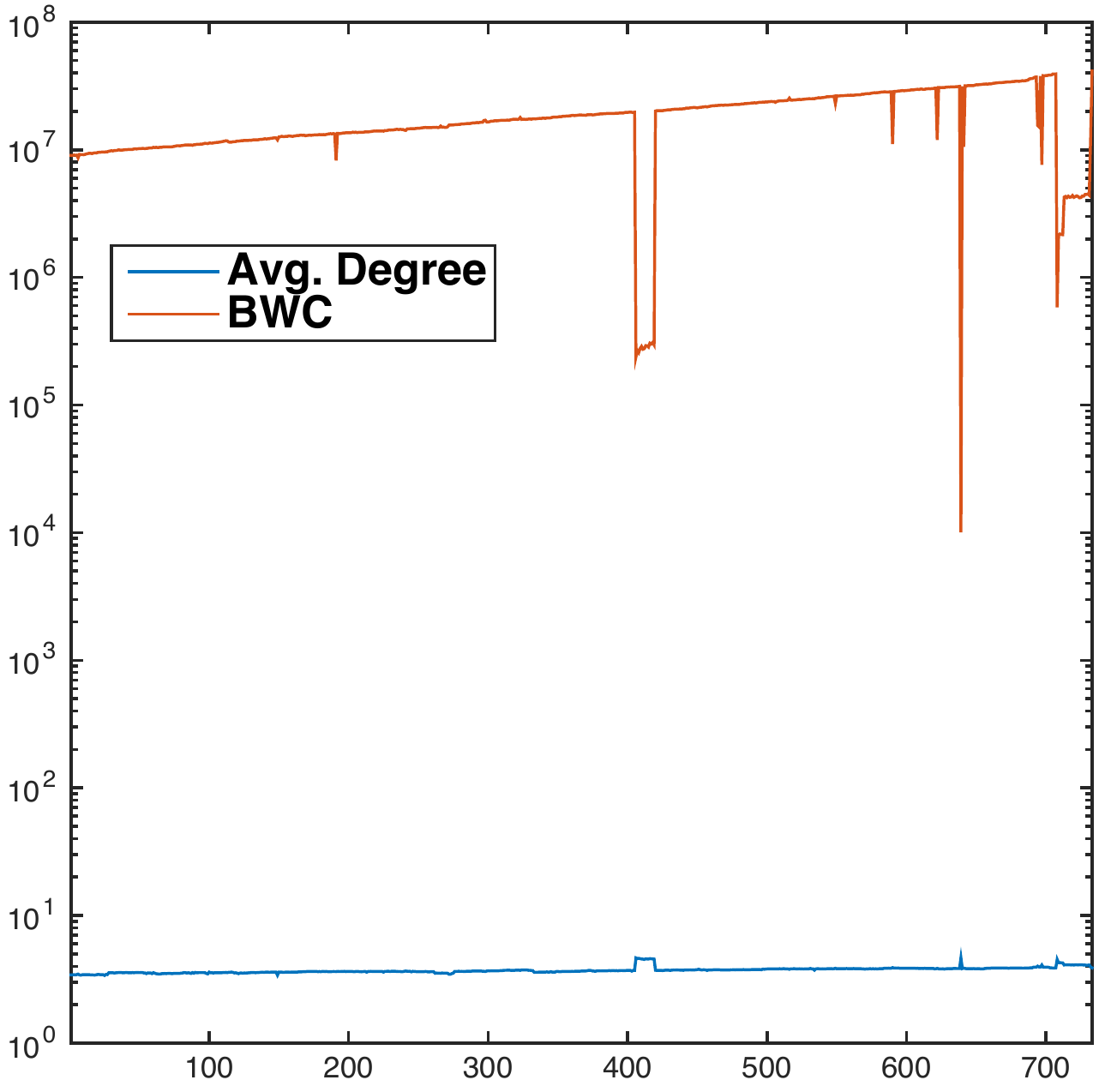}  &
\includegraphics[width=0.26\textwidth]{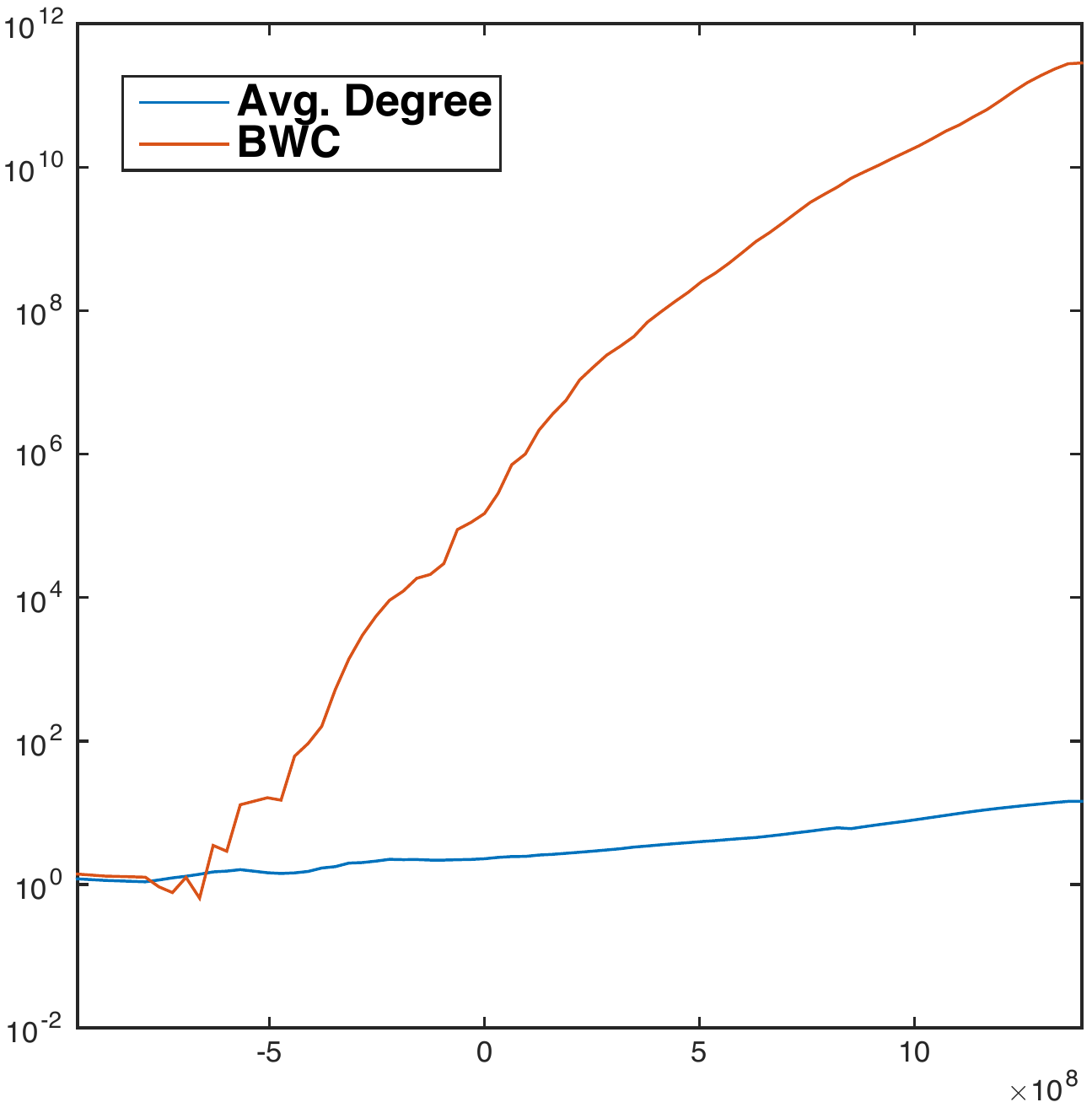} & 
\includegraphics[width=0.27\textwidth]{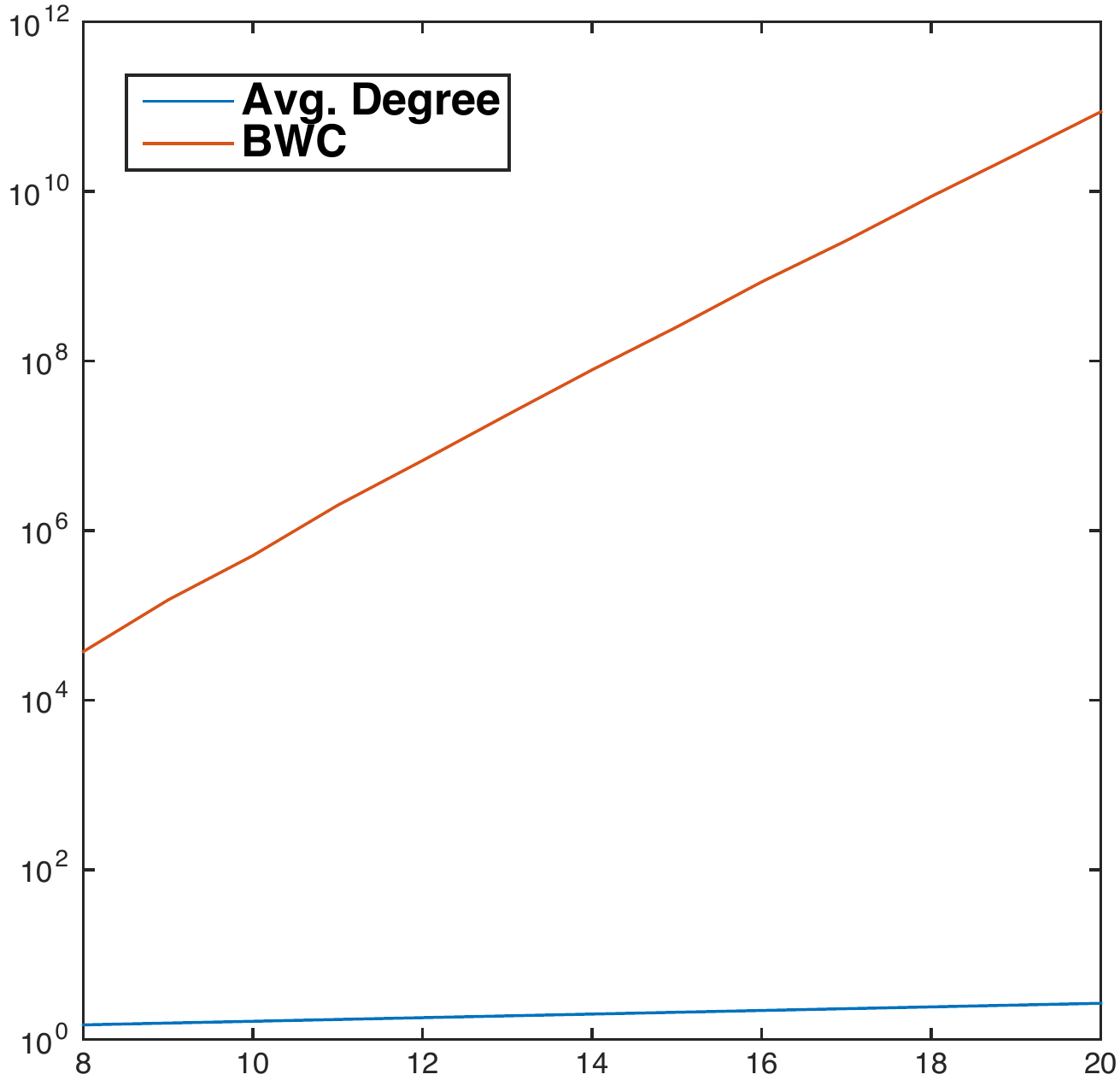} \\
(d) AS: $k=50$ & (e) DBLP: $k=50$ & (f) KG: $k=50$\\ 
\end{tabular}
\caption{\label{fig:timeevolving} \small Largest betweenness centrality score for $k=1$, and $k=50$,  number of nodes, edges and average degree versus time on the (i) 
\textsf{Autonomous systems}  (a),(d) (ii) \textsf{DBLP} dataset  (b),(e)
and (iii) stochastic Kronecker graphs (c),(f).}
\end{figure}

\begin{figure}[h]
\centering
\begin{tabular}{@{}c@{}@{\ }c@{}@{\ }c@{}}
\includegraphics[width=0.30\textwidth]{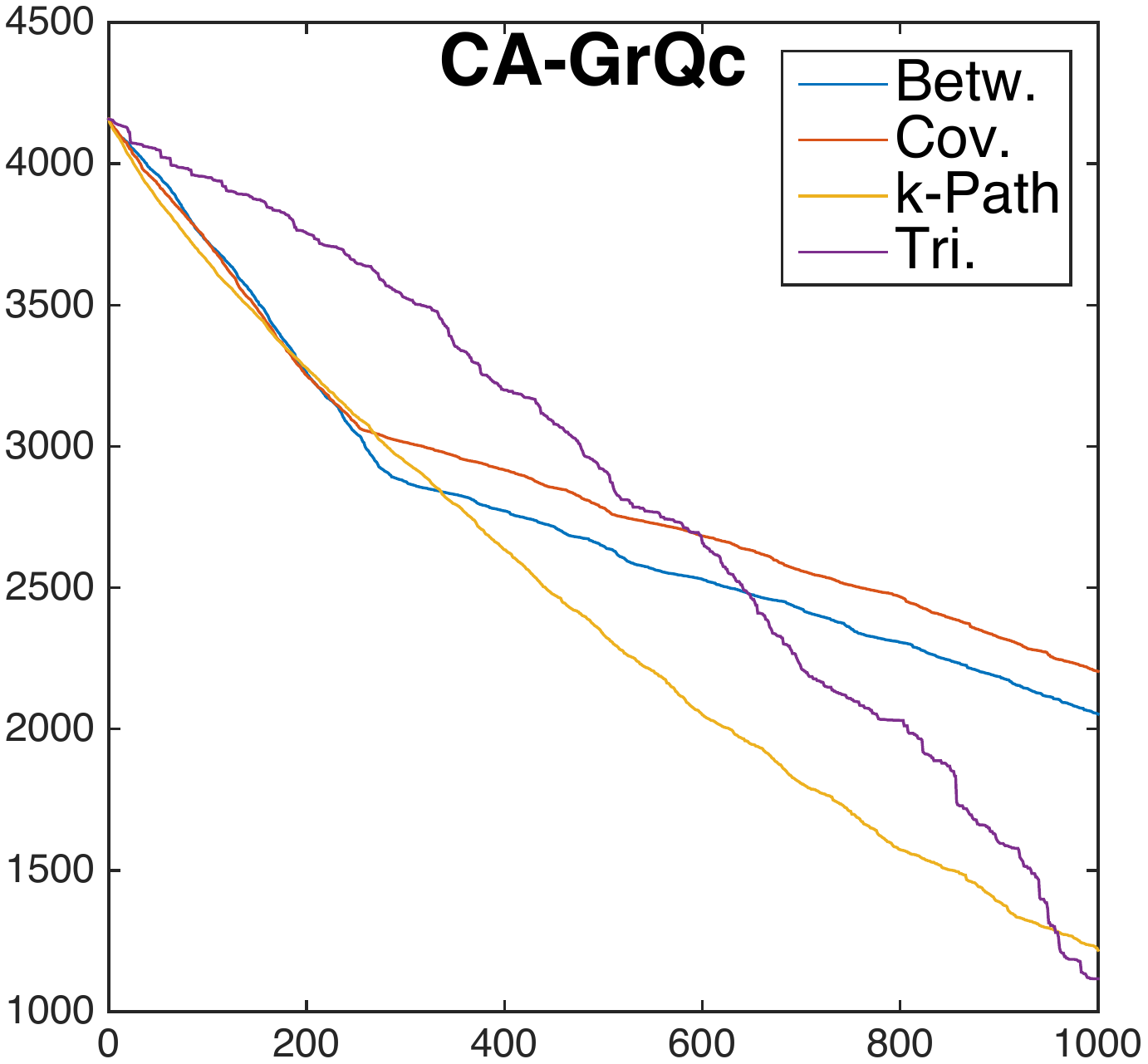} &
\includegraphics[width=0.30\textwidth]{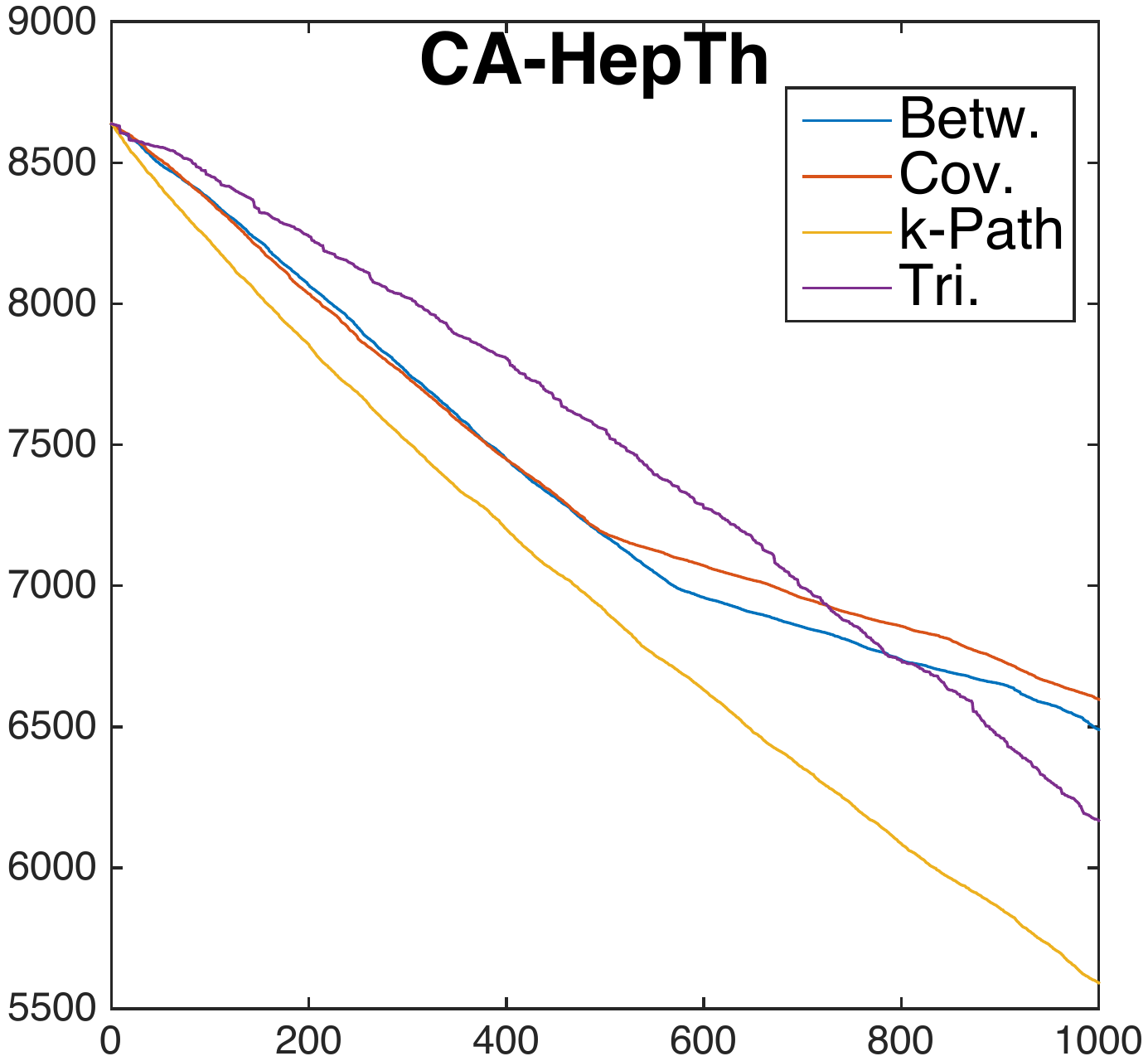} & 
\includegraphics[width=0.28\textwidth]{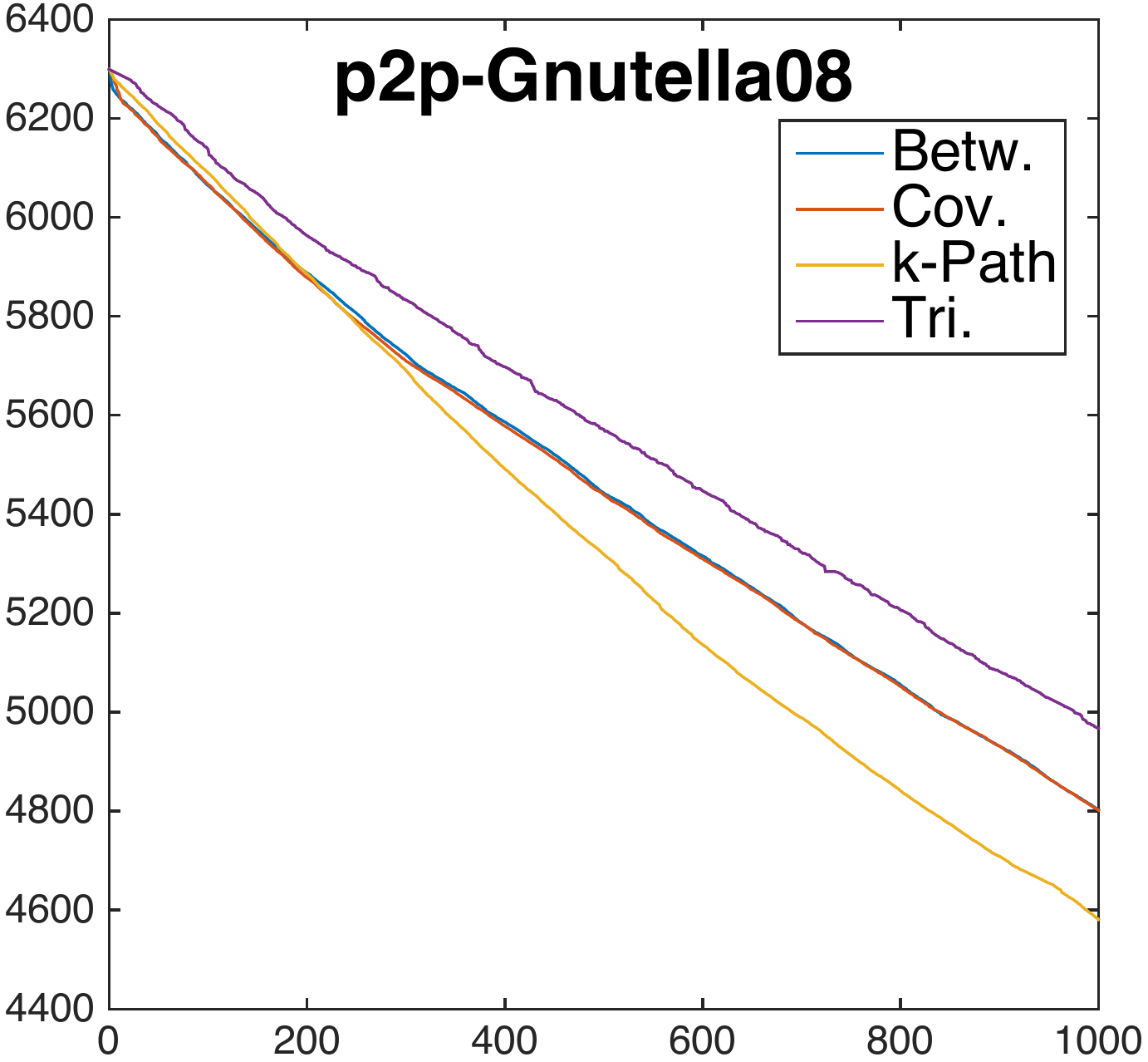} \\
(a) CA-GrQc & (b)CA-HepTh & (c) p2p-Gnutella08 \\  
 \includegraphics[width=0.30\textwidth]{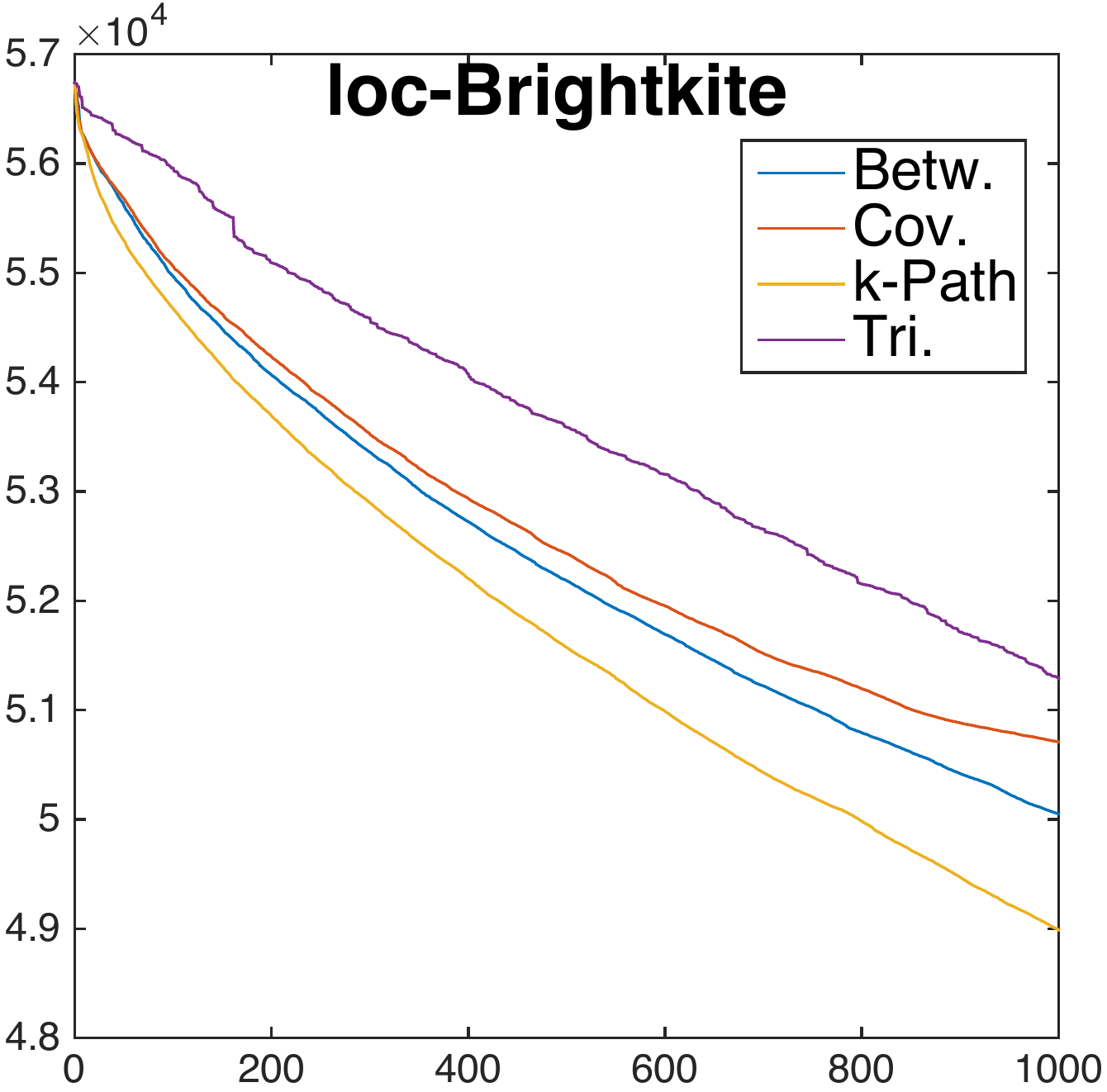}  &
\includegraphics[width=0.30\textwidth]{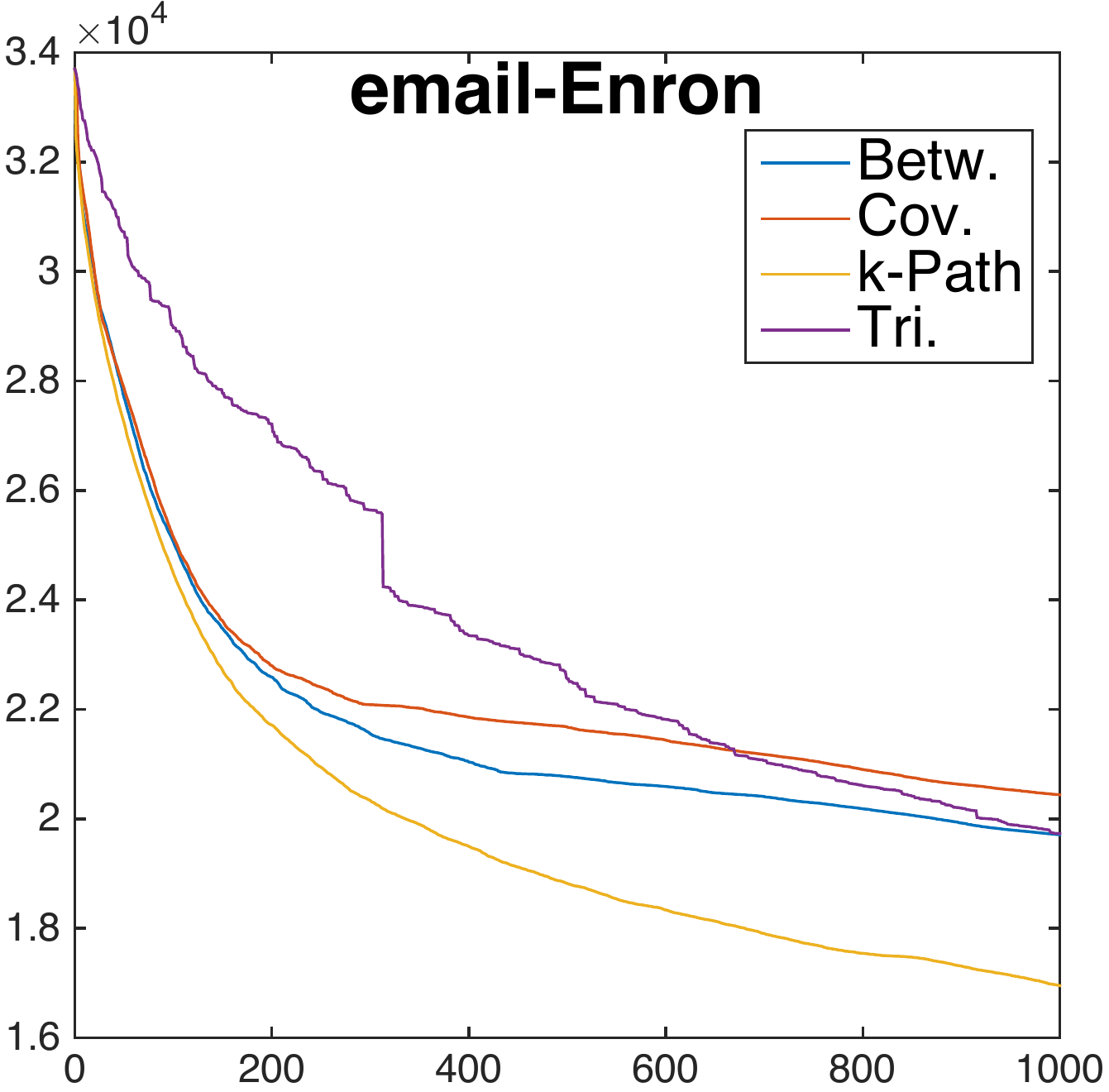} & 
\includegraphics[width=0.28\textwidth]{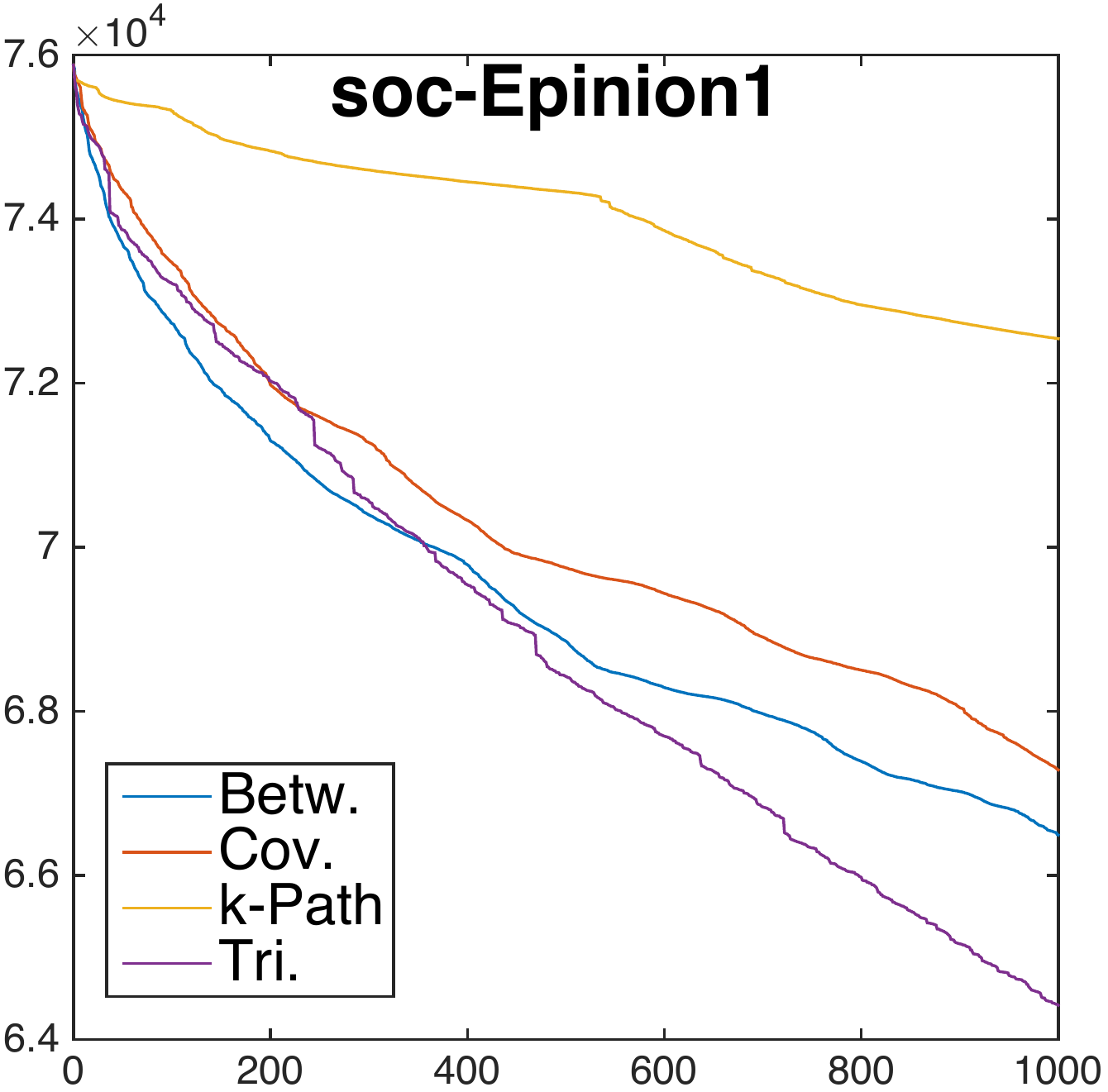} \\
(d) Brightkite& (e) Enron & (f) Epinions\\ 
\end{tabular}
\caption{\label{fig:attack}The size of the largest connected component, as we remove the first 1000 nodes in the order induced by centralities, for (a) CA-GrQc, (b)CA-HepTh, (c) p2p-Gnutella08, (d) Brightkite, (e) Enron, (f) Epinions. }
\end{figure}

\subsection*{Graph Attacks}

It is   well-known that scale-free networks are robust to random failures but vulnerable to targeted attacks \cite{albert2000error}. Some of the most efficient  strategies for attacking graph connectivity is based on removing iteratively the most central vertex in the graph \cite{holme2002attack}. Such sequential attacks are central in studying the robustness of the Internet and biological networks such as protein-protein interaction networks \cite{iyer2013attack}.

We remove the nodes one-by-one (according to the order induced by these centralities and picked by \algoname) and measure the size of the largest connected component. The results are plotted in Fig.~\ref{fig:attack}. Our observation is that  all the sizes of largest connected components decline significantly (almost linearly in size of $S$), which is compatible with our intuition of centralities. We also find that the $\kappa$-path and triangle centralities can be more effective at destroying the connectivity. Triangle centralities \cite{li2015triangle} can be  computed even on massive networks due to the efficient algorithms that have been developed for this problem, see, e.g., \cite{de2016tri,kolountzakis2012efficient,pagh2012colorful,seshadhri2013triadic,suri2011counting,tsourakakis2011triangle}.


\subsection*{Influence Maximization} 

We consider the Independent Cascade model~\cite{kempe2003maximizing}, where each edge has the probability $0.01$ of being active. For computing and maximizing the influence , we consider the algorithm of \cite{borgs2014maximizing} using $10^6$ number of samples (called {\he} but defined differently). We compute the influence of output of \algoname with output of \cite{borgs2014maximizing}. As shown in Table~\ref{table:inf}, and as we observe, the central nodes also have high influence, which shows a great correlation between being highly central and highly influential. It is worth outlining that our main point is to show that our proposed algorithm can be used to scale heuristic uses of \BW.

\input{tables/inf.tex}

%% file: tables/inf.tex
\begin{table}[H]
\centering
\small
\caption{\small Comparing the \emph{influence} of influential nodes (IM) and central nodes obtained by different centrality methods.}
\label{table:inf}

\begin{tabular}{lc|r|r|r|r|r|}
\cline{3-7}
                                                            & \multicolumn{1}{l|}{} & \multicolumn{5}{c|}{{\bf METHODS}}                                                                                                                                              \\ \hline
\multicolumn{1}{|l|}{{\bf GRAPHS}}                          & $k$                   & \multicolumn{1}{c|}{{\it IM}} & \multicolumn{1}{c|}{{\it betw.}} & \multicolumn{1}{c|}{{\it cov.}} & \multicolumn{1}{c|}{{\it $\kappa$-path}} & \multicolumn{1}{c|}{{\it tri.}} \\ \hline
\multicolumn{1}{|l|}{\multirow{3}{*}{{\bf CA-GrQc}}}        & {\it 10}              & 19.12                         & 13.67                            & 14.93                           & 14.10                                    & 18.48                           \\ \cline{2-7} 
\multicolumn{1}{|l|}{}                                      & {\it 50}              & 76.65                         & 67.28                            & 67.44                           & 65.06                                    & 69.30                           \\ \cline{2-7} 
\multicolumn{1}{|l|}{}                                      & {\it 100}             & 141.33                        & 126.76                           & 126.66                          & 124.51                                   & 124.06                          \\ \hline
\multicolumn{1}{|l|}{\multirow{3}{*}{{\bf CA-HepTh}}}       & {\it 10}              & 17.33                         & 15.61                            & 15.58                           & 14.63                                    & 12.98                           \\ \cline{2-7} 
\multicolumn{1}{|l|}{}                                      & {\it 50}              & 77.88                         & 70.53                            & 69.95                           & 67.80                                    & 63.95                           \\ \cline{2-7} 
\multicolumn{1}{|l|}{}                                      & {\it 100}             & 147.75                        & 133.45                           & 133.24                          & 130.41                                   & 127.52                          \\ \hline
\multicolumn{1}{|l|}{\multirow{3}{*}{{\bf p2p-Gnutella08}}} & {\it 10}              & 19.61                         & 13.05                            & 13.71                           & 10.39                                    & 18.06                           \\ \cline{2-7} 
\multicolumn{1}{|l|}{}                                      & {\it 50}              & 83.64                         & 60.58                            & 61.73                           & 51.57                                    & 74.19                           \\ \cline{2-7} 
\multicolumn{1}{|l|}{}                                      & {\it 100}             & 148.86                        & 118.27                           & 118.76                          & 103.58                                   & 132.04                          \\ \hline
\multicolumn{1}{|l|}{\multirow{3}{*}{{\bf email-Enron}}}    & {\it 10}              & 461.84                        & 458.70                           & 450.34                          & 455.25                                   & 451.53                          \\ \cline{2-7} 
\multicolumn{1}{|l|}{}                                      & {\it 50}              & 719.86                        & 703.08                           & 695.81                          & 699.74                                   & 681.05                          \\ \cline{2-7} 
\multicolumn{1}{|l|}{}                                      & {\it 100}             & 887.63                        & 863.66                           & 858.39                          & 865.76                                   & 830.15                          \\ \hline
\multicolumn{1}{|l|}{\multirow{3}{*}{{\bf loc-Brightkite}}} & {\it 10}              & 184.40                        & 162.64                           & 160.35                          & 163.16                                   & 145.19                          \\ \cline{2-7} 
\multicolumn{1}{|l|}{}                                      & {\it 50}              & 402.85                        & 372.64                           & 360.64                          & 366.28                                   & 330.45                          \\ \cline{2-7} 
\multicolumn{1}{|l|}{}                                      & {\it 100}             & 563.13                        & 521.18                           & 508.59                          & 512.77                                   & 445.11                          \\ \hline
\multicolumn{1}{|l|}{\multirow{3}{*}{{\bf soc-Epinion1}}}   & {\it 10}              & 343.89                        & 81.57                            & 111.47                          & 14.43                                    & 311.74                          \\ \cline{2-7} 
\multicolumn{1}{|l|}{}                                      & {\it 50}              & 846.18                        & 300.88                           & 282.88                          & 72.90                                    & 778.56                          \\ \cline{2-7} 
\multicolumn{1}{|l|}{}                                      & {\it 100}             & 1161.45                       & 463.04                           & 457.29                          & 133.20                                   & 1062.99                         \\ \hline
\end{tabular}
\end{table}

%% file: tex/Conclusion.tex
\section{Conclusion} 
\label{sec:concl}
In this work, we provide \algoname, a scalable algorithm for the (betweenness) Centrality Maximization problem, with theoretical guarantees. We also provide a general analytical framework for our analysis which can be applied to any monotone-submodular centrality measure that admits a {\he} sampler. We perform an  experimental analysis of our method on real-world networks which shows that our algorithm scales gracefully as the size of the graph grows while providing accurate estimations.  Finally, we study some interesting properties of the most central nodes.  

A question worth investigating is whether removing nodes in the reserve order, namely by starting from the least central ones, 
produces a structure-revealing permutation, as it happens with peeling in the context of dense subgraph 
discovery   \cite{Char00,tsourakakis2015k}.

\section{Acknowledgements}
This work was supported in part by NSF grant IIS-1247581 and NIH grant R01-CA180776.


